\pdfoutput=1
\documentclass[sigconf]{acmart}

\makeatletter                   
\def\mdseries@tt{m}             
\makeatother                    
\usepackage[plain]{fancyref}
\usepackage[draft=true]{minted} 
\usepackage{color}
\usepackage{hyperref}           
\hypersetup{
    colorlinks=true,
    linkcolor=blue,
    filecolor=red,      
    urlcolor=magenta,
    breaklinks=true,            
}
\usepackage{breakurl}           

\acmConference[]{}{}{}
\acmYear{2021}
\acmISBN{} 
\acmDOI{} 
\startPage{1}

\setcopyright{none}

\usepackage{booktabs}   
\usepackage{subcaption} 

\usepackage{amsmath, listings, amsthm, proof, xspace}
\usepackage{mathtools}
\usepackage{stmaryrd}
\usepackage{dsfont}

\usepackage{footmisc}
\usepackage{lipsum, babel}
\usepackage{soul}
\usepackage{multirow}
\usepackage[ruled,vlined]{algorithm2e}
\usepackage{tikz}
\usetikzlibrary{arrows,automata,positioning, shapes, shapes.geometric, fit, calc}
\usepackage{amsthm}
\usepackage{enumitem}

\usepackage{chngcntr}
\usepackage{apptools}
\AtAppendix{\counterwithin{theorem}{section}}

\newtheorem{theorem}{Theorem}

\newtheorem{definition}{Definition}
\newtheorem{example}{Example}

\newcommand{\alt}{~~|~~}

\newcommand{\binopdef}     \oplus 
\newcommand{\unopdef}      \ominus 

\newcommand{\ife}      [3] {\ifline{#1}~\thenline{#2}~\elseline{#3}}

\newcommand{\ifop}         {\mathsf{if}}
\newcommand{\ifline}   [1] {\ifop~ #1}
\newcommand{\thenline} [1] {\mathsf{then}~#1}
\newcommand{\elseline} [1] {\mathsf{else}~#1}

\makeatletter
\DeclareRobustCommand*\cal{\@fontswitch\relax\mathcal}
\makeatother

\def\cA{{\cal A}}

\def\cD{{\cal D}}
\def\cE{{\cal E}}

\def\cG{{\cal G}}

\def\cL{{\cal L}}

\def\cN{{\cal N}}

\newcommand{\tup} [1] {\langle #1 \rangle}

    \newcommand{\infral} [3] {\infer[\textsc{#3}]{\begin{array}{c} #2 \end{array} }{ \begin{array}{c} #1  \end{array} } }

\newcommand{\exec} [1] {\llbracket #1 \rrbracket}

\newcommand{\TODO}[1]{\hl{\textbf{TODO:} #1}\xspace}
\newcommand{\Comment} [1] {}

\newcommand{\proofblocks} [1] {}

\newcommand{\dref} [1] {}

\DeclareMathOperator*{\argmin}{arg\,min}

\begin{document}

\sloppy                         
\title[]{Program Synthesis Over Noisy Data with Guarantees} 


\author{Shivam Handa}
\affiliation{
  \department{EECS}              
  \institution{Massachusetts Institute of Technology}            
  \country{USA}                    
}
\email{shivam@mit.edu}          

\author{Martin Rinard}
\affiliation{
  \department{EECS}             
  \institution{Massachusetts Institute of Technology}           
  \country{USA}                   
}
\email{rinard@csail.mit.edu}         

\begin{abstract}
    We explore and formalize the task of synthesizing programs
over noisy data, i.e., data that may contain corrupted input-output
examples. 
By formalizing the concept of a Noise Source, an Input Source, and a prior
distribution over programs, we formalize the probabilistic process which
constructs a noisy dataset.
This formalism allows us to define the correctness of a synthesis algorithm, 
in terms of its ability to synthesize the
hidden underlying program. The probability of a synthesis algorithm being
correct depends upon the match between the 
Noise Source and the Loss Function used in the synthesis algorithm's
optimization process.  We formalize the concept of an optimal Loss
Function given prior information about the Noise Source.
We provide a technique to design optimal Loss Functions given perfect and
imperfect information about the Noise Sources.
We also formalize the concept and conditions required for convergence, i.e.,
conditions under which the probability that the synthesis algorithm produces
a correct program increases as the size of the noisy data set increases. 
This paper presents the first formalization of the concept of optimal Loss Functions, the
first closed form definition of optimal Loss Functions, 
and the first conditions that ensure that a noisy synthesis algorithm 
will have convergence guarantees. 

\end{abstract}



\maketitle

\section{Introduction}
Program synthesis has been successfully used to synthesize programs
from examples, for domains such as string transformations~\cite{gulwani2011automating, 
singh2016transforming}, 
data wrangling~\cite{feng2017component}, data completion \cite{wang2017synthesis}, and
data structure manipulation~\cite{feser2015synthesizing, 
osera2015type, yaghmazadeh2016synthesizing}. 
In recent years, there has been interest in synthesizing programs
from input-output examples in presence of noise/corruptions~
\cite{handa2020inductive, raychev2016learning, peleg2020perfect}. 
The motivation behind this line of work is to extract information left intact in
noisy/corrupted data to synthesize the correct program. 
These techniques 
have empirically shown that synthesizing the correct program, even in presence
of substantial noise, is possible.


\subsection{Noisy Program Synthesis Framework}
Previous work~\cite{handa2020inductive, peleg2020perfect} 
formulates the noisy program synthesis problem as an 
optimization problem over the program space and the noisy dataset.
Given a Loss Function, these techniques return the program which best fits the
noisy dataset. 
A {\bf Loss Function}, within this context, measures the cost of the 
input-output examples on which the program $p$ produces a different output 
than the output in the noisy dataset.

However, no previous research explores the connection 
between the best-fit program and the hidden underlying
program which produced (via the Noise Source) the
noisy data set. No previous research characterizes when the noisy program
synthesis algorithm, working with a given Loss Function, 
will synthesize a program equivalent to the hidden underlying program. 
Nor does it specify how to pick an appropriate Loss Function
to maximize the synthesis algorithm's probability of synthesizing the correct program.

We formulate the correctness of Loss Function based 
noisy program synthesis techniques in
terms of their ability to find the underlying hidden program. We achieve this by
formalizing and specifying the underlying hidden probabilistic process which
selects the hidden program, generates inputs, and 
corrupts the program's outputs to produce the noisy data set provided
to the synthesis algorithm. 
Our formalism uses the following concepts:
\begin{itemize}[leftmargin=*]
    \item {\bf Program Space:} A set of programs $p$ defined by a 
grammar $\cG$, 
\item {\bf Prior Distribution over Programs:} $\rho_p$ from which a {\it
    hidden underlying program} $p_h$ is sampled from,
\item {\bf Input Source:} A probability distribution $\rho_i$ which generates $n$
    input examples $\vec{x} = \tup{x_1, \ldots x_n}$,
\item {\bf Hidden Outputs:} The correct outputs computed by the hidden program
    $p_h$ over the input examples $\vec{x}$. 
\item {\bf Noise Source:} A probability distribution $\rho_N$, which corrupts
    the {\it hidden outputs} to construct a set of noisy outputs $\vec{y} =
        \tup{y_1, \ldots y_n}$, 
\item {\bf Noisy Dataset:} The collection of inputs and noisy outputs $\cD =
    (\vec{x}, \vec{y})$ which are visible to the synthesis algorithm to
        synthesize the {\it hidden underlying program} $p_h$. 
\end{itemize}

\noindent Working with this formalism, we present the following two results:
\begin{itemize}[leftmargin=*]
\item {\bf Optimal Loss Functions:} Given information about the Noise Source,
we formally define the optimal Loss Function and provide a technique to 
design the optimal Loss Function given perfect and imperfect information about the
Noise Source.
\item {\bf Convergence:} We formalize the concept of convergence and the conditions
required for the program synthesis technique to converge to the correct program 
as the size of the noisy data set increases. 
\end{itemize}

\Comment{
\subsection{The Synthesis Technique}
Handa et al. proposed a technique to synthesize programs over noisy datasets,
allowing for a large class of Noise Sources~\cite{handa2020inductive}.
Their technique however lack the ability to incorporate prior information about
the underlying hidden program.
A standard way to incorporate prior information about a hidden model (in this
case, a program) is to
introduce it via a prior probability distribution~\cite{gelman2013bayesian}.
Given a noisy dataset,
a prior distribution over programs allows us to direct our synthesis to take 
into account relative probabilities of programs.
Within this paper,
we propose a modification to the synthesis algorithm
in~\cite{handa2020inductive} to incorporate prior
probabilities over programs. This is the first technique which can solve the 
noisy synthesis problems for a large class of Noise Sources, with guarantees, 
and allows us to introduce prior information about the hidden program.

The synthesis algorithm is built upon the following concepts:
\begin{itemize}[leftmargin=*]
\item {\bf Loss Function:} A Loss Function $\cL(p,\cD)$ that measures the cost of the input-output examples on which
$p$ produces a different output than the output in the data set $\cD$, 
\item {\bf Prior Distribution over Programs:} A probability distribution over
    programs in DSL $\cG$ expressed via weights over terminals and production
        rules in $\cG$, 
\item {\bf Complexity Measure:} A complexity measure $C(p)$ 
    that measures the complexity
of a given program $p$. 
\end{itemize}
}
\Comment{
\TODO{Change THis}
Our technique uses a finite tree automata (FTA) to partition the space of programs
$p$ defined by the grammar $G$ into equivalence classes. Each equivalence class
consists of all programs with identical input-output behavior on the inputs in
the given dataset. These equivalence classes allow us to select a set of optimal
outputs which optimally trade-offs between its loss with respect to the noisy
outputs, and the prior probability of these outputs. Given the optimal outputs,
our technique selects a program which minimizes the complexity measure and
returns the optimal outputs over the given inputs.
}

\subsection{Optimal Loss Functions}
Different Loss Functions can be appropriate for different Noise Sources
and synthesis domains.
For example, the $0/1$ Loss Function, which counts the number of input-output
examples where the noisy dataset and synthesized program $p$ disagree,
is a general Loss Function that assumes that the Noise Source can corrupt the output
arbitrarily.
The Damerau-Levenshtein (DL) Loss Function~\cite{damerau1964technique} 
measures the edit difference
under character insertions, deletions, substitutions, and/or transpositions.
It can
extract information present in partially corrupted outputs and thus can be appropriate for measuring
loss when text strings are corrupted by a discrete Noise Source.
\Comment{
The $0/\infty$ Loss Function, which is $\infty$ unless $p$ agrees
with all of the input-output examples in the data set $\cD$, 
specializes our technique to work with a
Noise Source which will not corrupt output examples i.e.
the standard programming-by-example program synthesis scenario.
}

If the Noise Source and a Loss Function are a 
good match, the noisy synthesis is often able to extract
enough information left intact in corrupted outputs to synthesize the correct program.
The better the match between the Noise Source and Loss Function, the higher
the probablity that the noisy synthesis algorithm will produce a correct
program (i.e, a program equivalent to the underlying hidden program). 
Given a noisy dataset and a Noise Source, an optimal Loss Function has the
highest probability of returning a correct program. 
We formally define the concept of an optimal Loss Function and 
derive a closed-form expression for optimal Loss Functions (Section~\ref{sec:optimal}). 

We also derive optimal Loss Functions for Noise
Sources previously studied in the noisy synthesis literature~\cite{handa2020inductive}
(Section~\ref{sec:casestudies}). 
These case studies provide a theoretical explanation for the empirical results provided
in~\cite{handa2020inductive}. 
\Comment{
Given a dataset is corrupted by 1-Delete Noise Source, our 
Specifically, it explains the reasoning behind
1-Delete Loss Function's ability to tolerate more noise in the dataset, when it
is 
corrupted by the 1-Delete Noise Source, compared
to the Damerau-Levenshtein Loss Function.
It also provided a theoretical reasoning behind $n$-Substitution Loss Function's
inability to synthesize the correct program in the same experiment.
}
\Comment{
The experimental results presented in Section~\ref{sec:casestudies} showcase how
optimal Loss Functions require lower number of correct input-output examples, 
compared to Non-optimal Loss Functions,
to synthesize the correct program.
}

\subsection{Convergence}
Convergence is an important property which is well studied in the statistics and
machine learning literature~\cite{kearns1994introduction, gelman2013bayesian}. 
Convergence properties connect the size of the noisy
dataset to the probability of synthesizing a correct underlying program.
Given a synthesis setting (i.e. a set of programs, an Input Source, a Noise
Source, a prior distribution, and a Loss Function), we define a convergence property
that allows us to guarantee that the probability of the synthesis algorithm
producing the correct underlying program increases as the size of the noisy
dataset increases. 

This is the first paper to formulate the conditions required for the synthesis
algorithm to have convergence guarantees. 
Given an Input Source and a Noise Source, these conditions allow us to
select an appropriate Loss Function, which will allow our synthesis algorithm to
guarantee convergence.

\Comment{
\subsection{Experimental Results} 
We have implemented our technique and applied it to various programs in the 
SyGuS 2018 benchmark set~\cite{SyGuS2018}.
The results indicate that the technique is effective at solving program synthesis
problems over strings with modestly sized solutions.
For discrete Noise Sources studied in noisy program synthesis 
literature~\cite{handa2020inductive}, the experimental results follow our 
theoretical results in this paper.
}

\subsection{Contributions}
This paper makes the following contributions:
\begin{itemize}[leftmargin=*]
\item {\bf Formalism:} It presents and formalizes a framework for noisy
    program synthesis. The framework includes the concepts of a prior probability
        distribution over programs, an Input Source, and a Noise Source. 
        It then formalizes how these structured distributions interact to form a 
        probabilistic process which creates the Noisy Dataset. Given a noisy
        data set, this formalism allows us to define a correct solution for the 
        noisy synthesis problem over that data set. 

        \Comment{
\item {\bf Technique:} It presents an implemented technique for inductive
    program synthesis over noisy data which incorporates the concept prior probability
        distributions over programs. This technique allows us to introduce prior
        information about the underlying program into the synthesis routine.
    }

\item {\bf Optimal Loss Function:}
It presents a framework to calculate the expected reward associated with
        predicting a program, given a noisy dataset. Based on this framework,
        it presents a technique to design Loss Functions which are optimal
        , i.e.,
        have the highest probability of returning the correct solution to the
        synthesis problem.

\item {\bf Convergence:} It formalizes the concept of convergence, i.e.,
    for any probability tolerance $\delta$,
        there exists a threshold dataset size $k$, 
        such that, given a random noisy dataset of size $n \geq k$ generated by
        a hidden program $p$, the synthesis algorithm will synthesize a program
        equivalent to $p$ with probability greater than $\delta$.
        Based on this definition, this paper formulates conditions on the Input
        Source, the Loss Function, and the Noise Source which ensure
        convergence.

\item {\bf Case Studies:}
It presents multiple case studies highlighting Input Sources, Noise Sources, and
Loss Functions which break these conditions and thus make convergence
        impossible. It also proves that these conditions hold for some Noise Sources
        and Loss Functions studied in prior work, thus providing a theoretical 
        explanation for the reported empirical results.
\end{itemize}

\Comment{
This paper's contributions will allow us to formally analyze correctness
properties of techniques designed for noisy program synthesis. The formalism
provided here will allow us to compare different synthesis algorithms, and thus
allow us to improve synthesis techniques in the future.
}

\section{Synthesis over Noisy Data}
Within this section, we formalize a conceptual framework to view Noisy
program synthesis.
We also introduce a modified noisy synthesis algorithm 
to accommodate the concept of prior distribution of programs.

\subsection{Noisy Program Synthesis Framework}
We first define the programs we consider, how inputs to the program are 
specified, and the program semantics. Without loss of generality, we assume
programs $p$ are specified as parse trees in a 
{\bf domain-specific language (DSL)} grammar $\cG$.
Internal nodes represent function invocations; leaves are 
constants/0-arity symbols in the DSL. 
A program $p$ executes in an input $x$.  
$\llbracket p \rrbracket x$ denotes the
output of $p$ on input $x$ 
($\llbracket . \rrbracket$ is defined in Figure~\ref{fig:exec_sem}).

\begin{figure}
    \[
        \begin{array}{c}
            \begin{array}{cc}
            \infral{}
            {\llbracket c \rrbracket x \Rightarrow c}
            {(Constant)}
            &
            \infral{}
            {\llbracket t \rrbracket x \Rightarrow x(t)}
            {(Variable)}
            \end{array}
            \\
            \\
            \infral{
            \llbracket n_1 \rrbracket x \Rightarrow v_1 ~~~~~ 
            \llbracket n_2 \rrbracket x \Rightarrow v_2 ~~~~ \ldots ~~~~
            \llbracket n_k \rrbracket x \Rightarrow v_k
            }
            {\llbracket f(n_1, n_2, \ldots n_k) \rrbracket x  \Rightarrow
            f(v_1, v_2, \ldots v_k)}{(Function)}
    \end{array}
    \]
\vspace{-.1in}
    \caption{Execution semantics for program $p$}
\vspace{-.1in}
    \label{fig:exec_sem}
\end{figure}

All valid programs (which can be executed) are defined by a DSL 
grammar $\cG = (T, N, P, s_0)$ where:
\begin{itemize}[leftmargin=*]
    \item $T$ is a set of terminal symbols. These may include 
        constants and symbols which may change value depending on the input
        $x$.
    \item $N$ is the set of nonterminals that represent subexpressions in our
        DSL.
    \item $P$ is the set of production rules of the form  $s \rightarrow f(s_1,
        \ldots, s_n)$, where $f$ is a built-in function in the DSL and 
        $s, s_1, \ldots, s_n$ are non-terminals in the grammar.
    \item $s_0 \in N$ is the start non-terminal in the grammar.
\end{itemize}

We assume that we are given a black box implementation of each built-in function $f$
in the DSL.  In general, all techniques explored within this paper can be generalized to
any DSL which can be specified within the above framework. This is a standard
way of specifying DSLs in program synthesis literature~\cite{handa2020inductive,
wang2017program}

\begin{example}\label{ex:dsl}
    The following DSL defines expressions over input x, 
    constants 2 and 3, and addition and multiplication:
    \[
        \begin{array}{rcl}
            n &:=& x \alt n + t \alt n \times t; \\
            t &:=& 2 \alt 3;
        \end{array}
    \]
\end{example}
{\noindent \bf Notation:} 
We will use the notation $p(x)$ 
to denote $\llbracket p \rrbracket
x$ within the rest of the paper. Given a vector of input values $\vec{x} =
\tup{x_1, x_2, \ldots x_n}$, we use the notation $p[\vec{x}]$ to denote vector 
$\tup{p(x_1), p(x_2), \ldots p(x_n)}$.

\noindent We use the notation $G \subseteq \cG$ to denote a finite subset of
programs accepted by $\cG$.
Given an input vector $\vec{x}$,
we use the notation $G[\vec{x}]$ to denote the set of outputs produced by
programs in $G$. Formally,
\[
    G[\vec{x}] = \{p[\vec{x}]~\vert~p \in G\}
\]
Given a set of programs $G$, an input vector $\vec{x}$, and an output vector $\vec{z}$, we
use the notation $G_{\vec{x}, \vec{z}}$ to denote the set of programs in $G$,
which given input vector $\vec{x}$ produce the output $\vec{z}$. Formally,
\[
    G_{\vec{x}, \vec{z}} = \{p \in G~\vert p[\vec{x}] = \vec{z} \}
\]
\noindent Given two programs $p_1, p_2 \in G$, we use the notion $p_1 \approx p_2$ to
imply program $p_1$ is equivalent to $p_2$, i.e., for all $x \in X$, $p_1(x) =
p_2(x)$. Given an input vector $\vec{x}$, we use the notion $p_1
\approx_{\vec{x}} p_2$ to imply that program $p_1$ and $p_2$ have the same
outputs on input vector $\vec{x}$, i.e., $p_1[\vec{x}] = p_2[\vec{x}]$.

\noindent
Given a set of programs $G$,
we use the notation $G_{p_h}$ to denote the set of programs in $G$ which are
equivalent to program $p_h$. Formally, 
\[G_{p_h} = \{p \in G~\vert~p \approx
p_h\}\]
We use the notation ${G}^C_{p_h}$ to denote the set of programs in $G$ which
are not equivalent to program $p_h$. Formally, 
\[{G}^C_{p_h} = \{p \in G~\vert~p \napprox p_h\}\]

\noindent{\bf Prior Distribution over Programs:}
Given a set of programs $G$, let $\rho_p$ be a prior distribution over programs in $G$.
We assume the hidden underlying program $p_h$ is sampled from the prior
distribution with probability $\rho_p(p_h)$.
Prior distributions over models are a standard way to introduce prior knowledge
withing systems which learn from data~\cite{gelman2013bayesian}.
The prior distribution allows us to direct the synthesis procedure
by introducing information about the underlying process.
The synthesis algorithm use this distribution to trade off performance
over the noisy data set and 
synthesizing the most likely program. 

Within this paper, we assume that $\rho_p$ is expressed via a 
set $G \subseteq \cG$ and a weight function $w_\cG$ over the DSL $\cG$. 
We assume, we are given a weight function $w_\cG$, which assigns a weight to each
terminal $t \in T$, and each production $s \leftarrow f(s_1, \ldots s_k)$ is
$P$.

Given a program $p \in G$, $\rho_p(p)$ is defined as 
\[
    \rho_p(p) = \frac{w(s_0, p)}{\sum\limits_{p \in G}w(s_0, p)}
\]
where $w(p)$ is computed via the following recursive definition:
\[
    \begin{array}{rcl}
        w(t, t) &:=& w_G(t)\\
        w(s, f(e_1, \ldots e_n)) &:=& 
        \sum\limits_{pd \in P}
        w_G(pd)\times \prod\limits_{i=1}^n w(s_i, e_i)
    \end{array}
\]
where $pd$ are productions of the form $s \leftarrow f(s_1, \ldots s_n)$.

\noindent{\bf Input Source:}
An Input Source is a probabilistic process which generates the inputs provided
to the hidden underlying program. Formally, an Input Source is a probability
distribution $\rho_i$, from which $n \in \mathds{N}$ inputs $\vec{x} = \tup{x_1,
\ldots x_n}$ are sampled with probability $\rho_i(\tup{x_1,
        \ldots x_n}~\vert~p_h,~n)$. Where $p_h$ is hidden underlying program.
        
        \noindent{\it Note:} 
        Within this paper, we assume that
        inputs are independent of the hidden program $p_h$. We leave the
        exploration of idea that the inputs maybe selected to provide
        information about the hidden program for future work. 

\noindent{\bf Noise Source:}
A Noise Source $N$ is a probabilistic process which corrupts the correct
outputs returned by the hidden program to create the noisy outputs. Formally, 
a Noise Source $N$ is attached with a probability distribution $\rho_N$. Given a
hidden program $p_h$ and a set of outputs $\tup{z_1, \ldots z_n}$, 
the noisy outputs $\tup{y_1, \ldots y_n}$ are sampled from the probability
distribution $\rho_N$, with probability 
        \[\rho_N(\tup{y_1, \ldots y_n}~\vert~\tup{z_1, \ldots z_n})\]
\noindent{\it Note:} Within this paper, we will use the notation $N$ and
$\rho_N$ for a Noise Source interchangeably.

\noindent{\bf Noisy Dataset:}
A noisy dataset $\cD$ is a set of input values, denoted by $\vec{x}$ and noisy output
values, denoted by $\vec{y}$.
Within this paper, we assume the dataset $\cD = (\vec{x}, \vec{y})$ of size $n$ is
constructed by the following process:
\begin{itemize}[leftmargin=*]
    \item A hidden program $p_h$ 
        is sampled from $G$ with probability
        $\rho_p(p_h)$.
    \item $n$ inputs $\vec{x} = \tup{x_1, \ldots x_n}$
        are sampled from probability distribution $\rho_i(\tup{x_1,
        \ldots x_n}~\vert~n)$.    
    \item We compute outputs $\vec{z} = \tup{z_1, \ldots z_n}$, where $z_i = p_h(x_i)$.
    \item The Noise Source introduces noise by corrupting outputs $z_1,
        \ldots z_n$ to $\vec{y} = \tup{y_1, \ldots y_n}$ with probability
        \[\rho_N(\tup{y_1, \ldots y_n}~\vert~\tup{z_1, \ldots z_n})\]
\end{itemize}

\noindent{\bf Correctness:}
The goal of the synthesis procedure is to find a program $p$ such that $p$ is
equivalent to $p_h$, i.e., $p \approx p_h$. But
even in absence of noise, it maybe be
impossible to synthesize a program equivalent to $p_h$. 

\noindent
Therefore, similar to Noise-Free programming-by-example systems \cite{wang2017program, polozov2015flashmeta},
we relax the synthesis requirements to find a program $p$ such that
$p$ and $p_h$ have the same outputs on the input
vector $\vec{x}$ ($p \approx_{\vec{x}} p_h$), 
given input vector $\vec{x} = \tup{x_1, \ldots x_n}$. 
Note that, even in this relaxed
setting, the noise introduced by the Noise Source may make it impossible to
acquire enough information to synthesize the correct program.

In Section~\ref{sec:convergence}, we will tackle the harder problem of convergence, i.e., 
probability of synthesizing a program $p$, equivalent to the
hidden program $p_h$ ($p \approx p_h$)
will improve as we increase the size of
the noisy dataset.

\subsection{Loss Function}\label{sec:lossFunction}
Given a dataset $\cD = (\vec{x}, \vec{y})$ 
(inputs $\vec{x} = \tup{x_1, \ldots x_n}$ 
and outputs $\vec{y} = \tup{y_1, \ldots
y_n}$), 
and a program $p$, a {\bf Loss Function} 
$\cL(p, \cD)$
calculates how incorrect the program is with respect to the given
dataset. A Loss Function, only depends on outputs
$\vec{y}$, 
and the outputs of the program $p$ over inputs $\vec{x}$ 
(i.e., $\vec{z} = p[\vec{x}]$).
Given programs $p_1, p_2$, such
that for all $x \in \vec{x}$, $p_1(x) = p_2(x)$, 
then $\cL(p_1, \cD) = \cL(p_2, \cD)$. 
We can rewrite the Loss Function as 
\[
    \cL(p, \cD) = \cL(p[\vec{x}], \vec{y})
\]

\begin{definition}{\bf Piecewise Loss Function:}
A Piecewise Loss\\ Function $\cL(p, \cD)$ can be expressed in the following form
\[
    \cL(p, \cD) = \sum\limits_{i=1}^n L(p(x_i), y_i)
\]
    where $\cD = (\vec{x}, \vec{y})$, $\vec{x} = \tup{x_1, \ldots x_n}$, 
    $\vec{y} = \tup{y_1, \ldots y_n}$, and 
    $L(z, y)$ is the per-example Loss Function.
\end{definition}

\begin{definition}
{\bf $0/1$ Loss Function:} The $0/1$ Loss Function \\ 
$\cL_{0/1}(p, \cD)$ counts the number of
    input-output examples where $p$ does not agree with the data set $\cD$:
    \[
        \cL_{0/1}(p, \cD) = \sum\limits_{i=1}^n 1 ~~\mathrm{if}~~ (y_i \neq 
        p(x_i)) ~~\mathrm{else}~~ 0
    \]
\end{definition}


\subsection{Complexity Measure}
\label{sec:complexityMeasure}
As previously defined in literature~\cite{handa2020inductive},
given a program $p$, a\\ {\bf Complexity Measure} $C(p)$ ranks 
programs independent of the input-output examples in the 
dataset $\cD$.  
This measure is used to synthesize a simple program out of the current correct
candidate programs.
Formally, a complexity measure is a function
$C(p)$ that maps each program $p$ expressible in $G$ 
to a real number. The following $\mathrm{Cost}(p)$ complexity
measure computes the complexity of given program $p$ represented as a 
parse tree recursively as follows:
\[
\begin{array}{rcl}
\mathrm{Cost}(t) &=& \mathrm{cost}(t)\\
\mathrm{Cost}(f(e_1, e_2, \ldots e_k))  &=& \mathrm{cost}(f) + \sum\limits_{i = 1}^k \mathrm{Cost}(e_i) 
\end{array}
\]
where $t$ and $f$ are terminals and  built-in functions in our DSL $\cG$ respectively.
Setting $\mathrm{cost}(t) = \mathrm{cost}(f) = 1$ delivers a complexity measure 
$\mathrm{Size}(p)$ that computes the size of $p$.  

\subsection{Program Synthesis over Noisy Data}
We modify the synthesis algorithm presented in
\cite{handa2020inductive} to include the concept of prior distributions over programs.
The synthesis algorithm builds upon Finite Tree Automata.

\begin{definition}[\bf FTA]
    A finite tree automaton  (FTA) over alphabet $F$ is a tuple
    $\cA = (Q, F, Q_f, \Delta)$ where $Q$ is a set of states, $Q_f
    \subseteq Q$ is the set of accepting states, and $\Delta$ is a set of
    transitions of the form $f(q_1, \ldots, q_k) \rightarrow q$ where 
    $q, q_1, \ldots q_k$ are states, $f \in F$.
\end{definition}

\begin{figure}
    \[
        \begin{array}{c}
                \infral{
                    t \in T, \vec{c} = \tup{\exec{t}x_1, \ldots \exec{t}x_n}
                }
                {
                    q^{\vec{c}}_t \in Q 
                }
                {(Term)}
                \\
                \\
                \infral{
                    q^{\vec{c}}_{s_0} \in Q
                }
                {
                    q^{\vec{c}}_{s_0} \in Q_f 
                }
                {(Final)}
            \\
            \\
            \infral{
                s \leftarrow f(s_1, \ldots s_k) \in P,
                \{q^{\vec{c}_1}_{s_1}, \ldots q^{\vec{c}_k}_{s_k}\} \subseteq Q,
                \\
                \vec{c} = \tup{c_{o,1}, \ldots c_{o, n}}, c_{o, i} = \exec{f(
                c_{1, i}, \ldots c_{l, i}
                )}
            }
            {
                q^{\vec{c}}_s \in Q,  f(q^{\vec{c}_1}_{s_1}, \ldots 
                q^{\vec{c}_k}_{s_k}) \rightarrow q^{\vec{c}}_s \in \Delta
            }
            {(Prod)}
        \end{array}
    \]
    \caption{Rules for constructing a FTA $\cA = (Q, F, Q_f, \Delta)$
    given inputs $\vec{x} = \tup{x_1, \ldots x_n}$ 
    and DSL $\cG = (T, N, P, s_0)$.}
    \label{fig:wfta_rules}
\end{figure}

Given an input vector $\vec{x} 
= \tup{x_1, \ldots x_n}$, 
Figure~\ref{fig:wfta_rules} presents rules for constructing a FTA
that accepts all programs in grammar $\cG$. The alphabet 
of the FTA contains built-in functions within the DSL. The states 
in the FTA are of the form $q_s^{\vec{c}}$, where $s$ is a symbol (terminal or
non-terminal) in $\cG$ and $\vec{c}$ is a vector of values.
The existence of state $q^{\vec{c}}_s$ implies that there exists a 
partial program  
which can 
map $\vec{x}$ to values $\vec{c}$. Similarly, the existence of 
transition $f(q_{s_1}^{\vec{c}_1}, 
q_{s_2}^{\vec{c}_2} \ldots q_{s_k}^{\vec{c}_k}) \rightarrow
q_s^{\vec{c}}$ implies $\forall j \in [1, n].f(\vec{c}_{1, j}, 
\vec{c}_{2, j} \ldots \vec{c}_{k, j}) =
\vec{c}_j$. 

The $\mathsf{Term}$ rule states that if we have a terminal $t$ (either 
a constant in our language or input symbol $x$), execute it with the input
$x_i$ and construct a state $q_t^{\vec{c}}$ (where $\vec{c}_i = \llbracket t
\rrbracket x_i$). The $\mathsf{Prod}$  rule 
states that, if we have a production rule 
$f(s_1, s_2, \ldots s_k) \rightarrow s \in \Delta$, and 
there exists states $q_{s_1}^{\vec{c}_1}, q_{s_2}^{\vec{c}_2} \ldots
q_{s_k}^{\vec{c}_k} \in Q$, 
then we also have state $q_s^{\vec{c}}$ in the FTA and a transition 
$f(q_{s_1}^{\vec{c}_1}, q_{s_2}^{\vec{c}_2}, \ldots q_{s_k}^{\vec{c}_k}) \rightarrow
q_{s}^{\vec{c}}$.

The FTA $\mathsf{Final}$ rule (Figure~\ref{fig:wfta_rules}) marks all states 
$q^{\vec{c}}_{s_0}$ with start symbol $s_0$ as accepting states regardless of the
values $\vec{c}$ attached to the state. 


The FTA divides the set of programs in the DSL into subsets. 
Given an input set $\vec{x}$, all programs in a subset produce the same
outputs (based on the accepting state), i.e., 
if a program $p$ is accepted by the
accepting state $q^{\vec{c}}_{s_0}$, then $p[\vec{x}] = \vec{c}$.

In general, the rules in Figure~\ref{fig:wfta_rules} may result in a FTA which 
has infinitely many states. To control the size of the resulting FTA, we do not 
add a new state within the constructed FTA if the smallest 
tree it will accept is larger 
than a given threshold $d$. 
This results in a FTA which accepts all programs which are consistent with the 
input-output example but are smaller than the given threshold (it may
accept some programs which are larger than the given threshold but it will never
accept a program which is inconsistent with the input-output example).
This is a standard practice in the synthesis
literature~\cite{wang2017program, polozov2015flashmeta}.  

We denote the FTA constructed from DSL $\cG$, given input vector $\vec{x}$
 and threshold $b$  as
$\cA_\cG^d[\vec{x}]$. We omit the subscript grammar $\cG$ and threshold $b$
wherever it is obvious from context.

\begin{figure}[tb]
        \begin{tikzpicture}[->,>=stealth',shorten >=1pt,auto,node distance=1.4cm]

                        \tikzset{every state/.append style={rectangle}}
            \node[state, accepting, 
            initial by arrow, initial text={$x$}]           (1)  
            {$1$};
            \node[state, accepting]          (2)   [right of=1]    {$2$};
            \node[state, accepting]    (4) [above=1.0cm of 2]    {$4
            $};
            \node[state, accepting]    (3) [below=1.0cm of 2]    {$3
            $};
            \node[state, accepting]   (6) [right=1.0cm of 4] {$6$};
            \node[state, accepting]   (5) [below right=0.3cm and 1.4cm of 2]
            {$5$};
            \node[state, accepting]   (9) [below right=0.1cm and 1.4cm of 3]
            {$9$};
            \node[state, accepting]   (8)  [above=1.0cm of 4]  {$8$};
            \node[state, accepting]   (12) [left of=8]  {$12$};
            \node[state, accepting]   (7)  [right of=8] {$7$};

            \path 
                    (1)   edge node [sloped, below]{$+2, \times 3$} (3)
                        edge node [sloped, above] {$\times 2$}  (2)
                        edge node [sloped, above] {$+3$}  (4)
                  (2)   edge node [sloped, above] {$+2, \times 2$} (4)
                        edge node [sloped, above] {$\times 3$} (6)
                        edge node [sloped, above] {$+3$} (5)
                  (3)   edge node [sloped, above]{$\times 2, + 3$} (6)
                        edge node [sloped, above]{$+2$} (5)
                        edge node [sloped, above] {$\times 3$} (9)
                  (4)   edge node [right] {$\times 2$} (8)
                        edge node [sloped, above] {$+2$} (6)
                        edge node [left] {$\times 3$} (12)
                        edge node [right] {$+3$} (7)
            ;
        \end{tikzpicture}
    \caption{The FTA constructed for Example~\ref{ex:wfta1}}
    \label{fig:wfta1}
\end{figure}
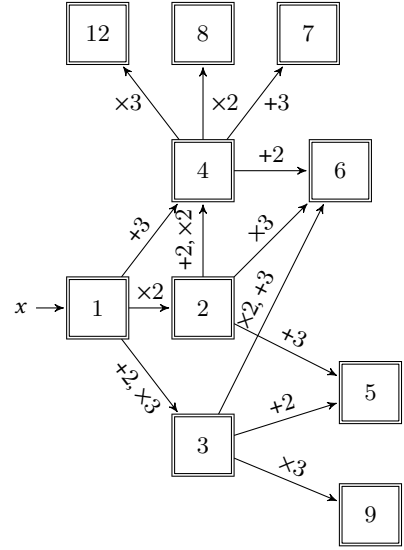

\begin{example}\label{ex:wfta1}
    Consider the DSL presented in Example~\ref{ex:dsl}.
    Given input $x=1$, 
    Figure~\ref{fig:wfta1} presents the
    FTA which represents all
    programs of height less than $3$.
\end{example}
\noindent For readability, we omit the states for terminals $2$ and $3$.

Algorithm~\ref{alg:regsyn} presents a modified version of the synthesis
algorithm presented by~\cite{handa2020inductive} to
synthesize programs within the noisy synthesis settings.
The algorithm first constructs a FTA over input vector $\vec{x}$. 
It then finds the state $q^{\vec{c}}_{s_0}$ which optimizes $\cL(\vec{c}, \vec{y}) 
- \log~
\pi(q^{\vec{c}}_{s_0}, \cG, \cA, d)$, where 
$\pi(q^{\vec{c}}_{s_0}, \cG, \cA, d)$ is defined in figure~\ref{fig:pi}.

$w(q^{\vec{c}}_s, m)$ denotes the sum of weights of partial programs of size
$\leq m$, accepted by the state $q^{\vec{c}}_s$.
$\pi(q^{\vec{c}}_{s_0}, \cG, \cA, d)$ denotes the sum of probabilities of all
programs accepted by state $q^{\vec{c}}_{s_0}$ of size $\leq d$. 
Note that, 
$\pi(q^{\vec{c}}_{s_0}, \cG, \cA, d) = \rho_p(G_{\vec{x}, \vec{c}})$, where $G$
is the subset of programs in $\cG$ of height $\leq d$.

\begin{figure*}
$
    \begin{array}{rcl}
        w(q^{\vec{c}}_t, m) &:=& w_\cG(t)\\
        w(q^{\vec{c}}_s, m) &:=& 
        \sum\limits_{f(q^{\vec{c}_1}_{s_1}, \ldots q^{\vec{c}_k}_{s_k}) 
        \rightarrow q^{\vec{c}}_s \in \Delta} 
        w_\cG(f)(\prod\limits_{i \in [1, k]}
        w(q^{\vec{c}_i}_{s_i}, m - 1) )\\
        \pi(q^{\vec{c}}_{s_0}, \cG, \cA, d) &:=& 
            w(q^{\vec{c}}_{s_0}, d)/
            \sum\limits_{q^{\vec{z}}_{s_0} \in Q_f}w(q^{\vec{z}}_{s_0}, d)
    \end{array}
$
        \caption{
            Rules for calculating $\pi(q^{\vec{c}}_{s_0}, \cG, \cA, d)$}
        \label{fig:pi}
    \end{figure*}

Note that $\vec{c}$ minimizes $\cL(\vec{c}, \vec{y}) - \log \rho_p(G_{\vec{x},
\vec{c}})$, where $G_{\vec{x}, \vec{c}}$ is the set of programs in $G$ which map
input $\vec{x}$ to output $\vec{c}$.

Given the optimal state $q^*$, we find the program 
accepted by $q^*$ which minimizes our complexity metric $C(p)$. 
The following equations hold true:
\[
    \vec{c} = \argmin_{\vec{c} \in G[\vec{x}]} \cL(\vec{c}, \vec{y}) 
    - \log \rho_p(G_{\vec{x}, \vec{c}}),~~~
    p^* = \argmin_{p \in G_{\vec{x}, \vec{c}}}C(p)
\]
\[
    p^* \in G_{\vec{x}, \vec{c}}
    \iff \vec{c} = p^*[\vec{x}] \text{ and } p^* \in G
\]

 \begin{algorithm}
     \SetAlgoLined
     \SetKwInOut{Input}{Input}
     \Input{DSL $\cG$, prior distribution $\pi$, 
     threshold $d$, data set $\cD = (\vec{x}, \vec{y})$, 
     Loss function $\cL$, 
     complexity measure $C$}
\KwResult{Synthesized program $p^*$}
     $\cA^\cG_d[\vec{x}] = (Q, F, Q_f, \Delta)$\\ 
     $q^* \leftarrow \mathsf{argmin}_{q^{\vec{c}}_{s_0} 
     \in Q_f}~ \cL(\vec{c}, \vec{y}) - \log \pi(q^{\vec{c}}_{s_0}, \cG, \cA, d)$\\
     $p^* \leftarrow \mathsf{argmin}_{p \in (Q, F, \{q^*\}, \Delta)} C(p)$\\
 \caption{Synthesis Algorithm}
     \label{alg:regsyn}
 \end{algorithm}

Given a FTA $\cA$, we can use dynamic programming to find the 
minimum complexity parse tree 
(under the certain $\mathrm{Cost}(p)$ like measures) 
accepted by $\cA$~\cite{gallo1993directed}. 
In general, given a FTA $\cA$, we assume 
we are provided with a method to find the program $p$ accepted by $\cA$ which
minimizes the complexity measure.

\section{Optimal Loss Function}
\label{sec:optimal}
Within this section, we will formalize the connection between the Noise Source
and the Loss Function. We then derive a closed form expression for the optimal
Loss Function in case of perfect and imperfect information about the Noise
Source.

\subsection{Optimal Loss Function given a Noise Source}
\label{subsec:optimalloss}

\Comment{
Given dataset $\cD = (\vec{x}, \vec{y})$, we assume it was constructed by 
the following underlying process:
\begin{itemize}
    \item A hidden program $p_h$ is sampled from the set $G$ with probability
        distribution $\rho_p$.
    \item $n$ inputs $\vec{x} = \tup{x_1, \ldots x_n}$
        are sampled from probability distribution $\rho_i(\tup{x_1,
        \ldots x_n}~\vert n)$.   
    \item The process then computes outputs 
        $\vec{z} = \tup{z_1, \ldots z_n}$, where $z_i = p_h(x_i)$.
    \item A Noise Source $N$, introduces noise by corrupting outputs $z_1,
        \ldots z_n$ to $\vec{y} = \tup{y_1, \ldots y_n}$ with probability
        \[\rho_N(\tup{y_1, \ldots y_n} \vert \tup{z_1, \ldots z_n})\]
 \end{itemize}
Given $G$, $x_1, \ldots x_n$, and $y_1, \ldots y_n$, we want to
        synthesize program $p \approx p_h$. This problem is hard even in the
        absence of noise, i.e., when the dataset $\cD = \cD_c = (\vec{x},
        \vec{z})$. Even the correct dataset $\cD_c$ may not contain enough
        information to . Therefore, similar to previous work in non-noisy inductive
        program synthesis settings, we are instead search for a program $p$
        which has the same output as program $p_h$ on input set $\vec{x}$
        (denoted by $p \approx_{\vec{x}} p_h$). Formally, $p \approx_{\vec{x}}
        p_h$ if and only if $p[\vec{x}] = p_h[\vec{x}]$.
    }

Given dataset $\cD = (\vec{x}, \vec{y})$, let us  
assume that a synthesis
procedure predicts $p$ to be the underlying hidden program. The {\bf Expected Reward} is the probability
that $p$ generated dataset $\cD$.
Formally, given $\cD = (\vec{x},
\vec{y})$, 

\[\cE(p \vert \vec{x}, \vec{y}) = \sum\limits_{p_h \in G} \mathds{1}(p
\approx_{\vec{x}} p_h) 
\rho(p_h \vert \vec{x}, \vec{y})
\]
\[
    = \frac{1}{\rho(\vec{y}~\vert~\vec{x})}\sum\limits_{p_h \in G} \mathds{1}(p \approx_{\vec{x}} p_h) 
    \rho_p(p_h \vert \vec{x}) \rho_N(\vec{y} \vert p_h[\vec{x}])
\]
\[
    = \frac{1}{\rho(\vec{y}~\vert~\vec{x})}\sum\limits_{p_h \in G} \mathds{1}(p \approx_{\vec{x}} p_h) 
     \rho_p(p_h) \rho_N(\vec{y} \vert p_h[\vec{x}])
\]
\[
    \cE(p~\vert~\vec{x}, \vec{y}) = \frac{1}{\rho(\vec{y}~\vert~\vec{x})}
    \rho_p(G_{\vec{x}, p[\vec{x}]})\rho_N(\vec{y}~\vert~p[\vec{x}])
\]
Since the above reward only depends on the output of program $p$ on input set
$\vec{x}$, 
we can rewrite the above reward as
\[
    \cE(\vec{c}~\vert~\vec{x}, \vec{y}) = \frac{1}{\rho(\vec{y}~\vert~\vec{x})}\rho_p(G_{\vec{x},
    \vec{c}})\rho_N(\vec{y}~\vert~\vec{c})
\]
where $\vec{c} = p[\vec{x}]$. 

Therefore, given dataset $\cD$, 
the probability $\cD$ was generated by the synthesized program $p$ is
$\cE(p[\vec{x}]~\vert~\vec{x},\vec{y})$.

Given a Loss Function $\cL$ and prior distribution $\rho_p$, our synthesis algorithm is
correct with probability $\cE(p_l[\vec{x}]~\vert~\vec{x}, \vec{y})$, where $p_l$
is:
\[
    \vec{c} = \argmin_{\vec{c} \in G[\vec{x}]} \cL(\vec{c}, \vec{y}) - \log
    \rho_p(G_{\vec{x},
    \vec{c}}), ~~~
    p_l = \argmin_{p \in G_{\vec{x}, \vec{c}}} C(p)
\]
$p_l$ will be the ideal prediction if it maximizes the expected reward.
Let $\cE_L(p[\vec{x}]~\vert~\vec{x}, \vec{y}) = - \log \cE(p[\vec{x}]~\vert~\vec{x}, \vec{y})$.
\[
    \cE_L(p[\vec{x}]~\vert~\vec{x}, \vec{y}) = - \log \rho_p(G_{\vec{x},
    p[\vec{x}]}) + (- \log \rho_N(\vec{y}~\vert~p[\vec{x}])) + C
\]

Therefore, given a set of programs $G$, dataset $\cD$,
prior distribution $\rho_p$, and Noise Source $N$, 
and no other information
about the hidden program $p_h$, 
the synthesis algorithm will always return the program which maximizes the
expected reward if the Loss
Function $\cL(\vec{c}, \vec{y})$ is equal to 
$ - \log
\rho_N(\vec{y}~\vert~\vec{c}) + C$, for any constant $C$.
Hence, within this setting, $-\log \rho_N(\vec{y}~\vert~\vec{c}) + C$ is the optimal
Loss Function.

Note that, 
in expectation, no other Loss Function 
will out perform the above optimal version. 

\subsection{Optimal Loss Function given imperfect information}
Let us now consider a scenario where we are presented with some 
imperfect information about the Noise Source, i.e., the Noise Source corrupting
the correct output belongs to the set $\cN$ and we are presented with a prior
probability distribution $\rho_\cN$ over possible Noise Sources in $\cN$.
The probability
that $N \in \cN$ is 
the underlying Noise Source corrupting the correct dataset is $\rho_\cN(N)$.

Given dataset $\cD = (\vec{x}, \vec{y})$, we assume it was constructed by 
the following underlying process:
\begin{itemize}[leftmargin=*]
    \item A Noise Source $N \in \cN$ is sampled from the prior distribution
        $\rho_\cN$,
        with probability $\rho_\cN(N)$.
    \item A hidden program $p_h$ is sampled from the set $G$ with probability
        distribution $\rho_p$.
    \item $n$ inputs $\vec{x} = \tup{x_1, \ldots x_n}$
        are sampled from probability distribution $\rho_i$ with probability $\rho_i(\tup{x_1,
        \ldots x_n}~\vert n)$. 
    \item The process then computes outputs 
        $\vec{z} = \tup{z_1, \ldots z_n}$, where $z_i = p_h(x_i)$.
    \item The sampled Noise Source $N$, introduces noise by corrupting outputs $z_1,
        \ldots z_n$ to $\vec{y} = \tup{y_1, \ldots y_n}$ with probability
        \[\rho_C(\tup{y_1, \ldots y_n} \vert \tup{z_1, \ldots z_n}, N)\]
        which is equal to  
        \[
            \rho_N(\tup{y_1, \ldots y_n}~\vert~\tup{z_1, \ldots z_n})
        \]
 \end{itemize}

\noindent{\bf Expected Reward:} Given dataset $\cD$, 
let us  
assume that a synthesis
procedure predicts $p$ to be the underlying hidden program. The {\bf Expected Reward} is the probability
dataset $\cD$ was generated by $p$ as the hidden underlying program.
Formally, given $\cD = (\vec{x},
\vec{y})$, 

\[\cE(p \vert \vec{x}, \vec{y}) = \sum\limits_{p_h \in G} \mathds{1}(p
\approx_{\vec{x}} p_h) 
\rho(p_h \vert \vec{x}, \vec{y})
\]
\[
    \propto
    \sum\limits_{N \in \cN} \sum\limits_{p_h \in G} \mathds{1}(p \approx_{\vec{x}} p_h) 
    \rho_p(p_h \vert \vec{x}) \rho_C(\vec{y} \vert p_h[\vec{x}], N) \rho_\cN(N)
\]
\[
    = \sum\limits_{N \in \cN} \sum\limits_{p_h \in G} \mathds{1}(p \approx_{\vec{x}} p_h) 
     \rho_p(p_h) \rho_N(\vec{y} \vert p_h[\vec{x}], N) \rho_\cN(N)
\]
\[
    = \rho_p(G_{\vec{x}, p[\vec{x}]})
(
\sum\limits_{N \in \cN}
    \rho_C(\vec{y}~\vert~p[\vec{x}], N)\rho_\cN(N)
)
\]
Let $\cE_L(p[\vec{x}]~\vert~\vec{x}, \vec{y}) = - \log
\cE(p[\vec{x}]~\vert~\vec{x}, \vec{y})$
\[
    = - \log \rho_p(G_{\vec{x},
    p[\vec{x}]}) + (- \log (\sum\limits_{N \in \cN} \rho_\cN(N)
    \rho_C(\vec{y}~\vert~p[\vec{x}])))
\]
%

Therefore, given a set of programs $G$, dataset $\cD$,
prior distribution $\rho_p$,  
and no other information
about the hidden program $p_h$ or the hidden Noise Source, 
the synthesis algorithm will always return the 
program which maximizes the expected reward 
if the Loss
Function $\cL(\vec{c}, \vec{y})$ is equal to 
\[ - [\log
\sum\limits_{N \in \cN} \rho_C(\vec{y}~\vert~\vec{c}, N) \rho_\cN(N)] + C
   = -\log E[\rho_N(\vec{y}~\vert~\vec{c})] + C
\]
for any constant $C$.
Hence, the optimal Loss function, in presence of imperfect information of the
Noise Source, is  the negative $\log $ of the expected probability of output
$\vec{c}$ being corrupted to noisy output $\vec{y}$.

\section{Convergence}
\label{sec:convergence}
Within this section, we explore the conditions under which the synthesis
algorithm will have convergence guarantees.

Given a synthesis setting (i.e., a finite set of programs $G$, an Input Source
$\rho_i$, a Noise Source $\rho_N$, a prior probability $\rho_p$, and a Loss
Function $\cL$), convergence property allows us to guarantee synthesis
algorithm's output will be the correct underlying program with high probability
if we are providing the algorithm with a large enough dataset.

Given a Noise Source $N$, a Loss Function $\cL$, 
prior probability $\rho_p$,
a positive probability hidden program $p_h$ (i.e., $\rho_p(p_h) > 0$), 
and a dataset size $n$,  
let $Pr[p^n_s \approx p_h~\vert~p_h, N]$ denote the probability on synthesizing
program equivalent $p_h$ on a random data set $(\vec{x}, \vec{y})$ of size $n$, constructed
assuming $p_h$ as the hidden program.
Formally, $Pr[p^n_s \approx p_h~\vert~p_h, N]$ is the probability of
the following process returning true.
\begin{itemize}[leftmargin=*]
    \item Sample $n$ inputs $\vec{x}$ with probability $\rho_i(\vec{x}~\vert~n)$.
    \item $\vec{z}_h = p_h[\vec{x}]$, $\vec{y}$ is sampled from the distribution
        $\rho_N(\vec{y}~\vert~\vec{z}_h)$.
    \item Return true, if for all programs $p \in G^C_{p_h}$:
        \begin{equation*}
            \begin{split}
            \cL(p_h[\vec{x}],\vec{y}) - \log \rho_p(G_{\vec{x}, p_h[\vec{x}]}) < 
                \cL(p[\vec{x}],\vec{y}) - \log \rho_p(G_{\vec{x}, p[\vec{x}]}) 
            \end{split}
        \end{equation*}
        \noindent
        and there exists a program $p^n_s \in G_{p_h}$, such that for all $p \in
        G^C_{p_h}$, where $p \approx_{\vec{x}}
        p_h$,
        $
            C(p^n_s) < C(p)
        $.
\end{itemize}
Note that if the procedure returns true then the following is true:
        \[
            p_h[\vec{x}] = \argmin\limits_{\vec{z} \in G[\vec{x}]} \cL(\vec{z}, \vec{y}) -
            \log \rho_p(G_{\vec{x}, \vec{z}}),~~
            p^n_s = \argmin\limits_{p \in G_{\vec{x}, p_h[\vec{x}]}} C(p)
        \]
i.e., $p^n_s \approx p_h$ is synthesized.

$Pr[p^n_s \approx p_h~\vert~p_h, N]$ measures the probability that a 
program equivalent to the hidden program is synthesized, given a random noisy dataset
using $p_h$ as the underlying program.

\begin{definition}{\bf Convergence:}
Given a finite set of programs $G$, a Loss Function $\cL$,
    a Noise Source $N$, an Input Source $\rho_i$, and a probability distribution
    over programs $\rho_p$, 
    the synthesis will converge, 
    if for all $\delta > 0$, there exists a natural number $k$, such that for
    all positive probability hidden program $p_h$ and $n \geq k$:
    \[
        Pr[p^n_s \approx p_h~\vert~p_h, N] \geq (1- \delta)
    \]
    i.e., for all probability tolerance $\delta > 0$, we can find a minimum
    dataset size $k$, such that for all hidden programs and dataset sizes $\geq
    k$, the
    probability the synthesized program will be equivalent to the hidden program
    is greater that $(1 - \delta)$.
\end{definition}

If the convergence property is true for a 
finite set of programs $G$, a Loss Function $\cL$,
    a Noise Source $N$, an Input Source $\rho_i$, and a probability distribution
    over programs $\rho_p$, 
then for all $\delta > 0$,
there exists a natural number $k$, such that for
all positive probability programs $p_h$ (i.e., there is positive probability
that $p_h$ will be sampled), 
for any dataset of size greater than $k$, with probability $(1- \delta)$, the
algorithm will
synthesize a program equivalent to $p_h$.

\subsection{Differentiating Input Distributions}
Even in absence of noise, the Input Source may 
hinder the ability of the synthesis algorithm to converge to the hidden program.
For example, consider an Input Source which only generates vectors $\vec{x}$,
such that, there exist two programs $p_1$ and $p_2$ which have the same outputs
on input $\vec{x}$ (i.e., $p_1[\vec{x}] = p_2[\vec{x}]$).
For such an Input Source, the synthesis algorithm, even in the absence of noise,
cannot differentiate between datasets produced assuming $p_1$ is the underlying
program and datasets produced assuming $p_2$ is the underlying program.
This issue comes up in traditional Noise-Free programming-by-example synthesis
as well. Noise-Free synthesis frameworks assume that the Input Source will
eventually provide input-output examples to differentiate the underlying program
from all other possible candidate programs to guarantee convergence.
We take a similar approach and constrain the Input Source to provide convergence
guarantees.

Let $d$ be some distance metric which measures distance
between two equally sized output vectors. 
For any underlying program and a probability tolerance, increasing the dataset size
should eventually allow us to differentiate between this program and any other
program in $G$.

\begin{definition}{\bf Differentiating Input Source:}
    Given a set of programs $G$ and a distance metric $d$, an Input Source
    $\rho_i$ is  
\\
    {\bf differentiating}, if for a
    large enough dataset size, the distribution will return an input
    dataset which will differentiate any two programs within 
    $G$, with a high probability.
\end{definition}

An Input Source is {\bf differentiating} if 
for all $\delta > 0$ and $\epsilon > 0$, 
there exists a minimum dataset size $k$, 
such that for dataset sizes $n \geq k$
and all programs $p_h \in G$,
the following process returns
true with probability greater than $(1 - \delta)$:
\begin{itemize}[leftmargin=*]
    \item Sample $\vec{x}$ of size $n$ from the distribution $\rho_i(\vec{x})$. 
    \item Return true if $\forall p \in \tilde{G}_{p_h}$, $d(p[\vec{x}],
        p_h[\vec{x}]) \geq \epsilon$.
\end{itemize}

Formally, given a set of programs $G$ and a distance metric $d$, an Input Distribution
    $\rho_i$ is  {\bf differentiating}, if 
for all $\delta > 0$, for all distance $\epsilon > 0$, 
there exists a natural number $k$, such that for all natural numbers $n \geq k$,
and for all programs $p_h \in G$, the following statement is true:
\[
    \int\limits_{\vec{x} \in X^n} \mathds{1}(
    \forall p \in G^C_{p_h}.~d(p_h[\vec{x}], p[\vec{x}]) 
    > \epsilon)\rho_i(d \vec{x}) \geq (1 - \delta)
\]
i.e., When sampling input vectors $\vec{x}$ of length $n$, with probability 
greater than $(1 - \delta)$, the distance
between $p_h[\vec{x}]$ (output of program $p_h$ on input $\vec{x}$) and
$p[\vec{x}]$, for all programs $p$ not equivalent to $p_h$, 
is greater than $\epsilon$.

Having a differentiating Input Source ensures that as we increase the size of
the dataset, a random dataset will contain inputs which will allow us to
differentiate a program with other programs with high probability.

\subsection{Differentiating Noise Sources}
Even if we are given an Input Source which allows us to differentiate between
the hidden underlying program and other programs in $G$ in the absence of
noise, the Noise Source can, in theory, make convergence impossible.
For example, consider a Noise Source which corrupts all outputs $z$ to a single
noisy output value $y^*$.
A dataset corrupted by this Noise Source contains no information about the
underlying correct outputs. A synthesis process cannot extract any
information about the hidden underlying program from such a dataset.
Therefore, no
synthesis algorithm will be able to differentiate between different programs 
in the input program space.
Restrictions have to placed on the types of Noise Sources a synthesis algorithm
can handle in order to provide convergence guarantees.

\begin{definition}{\bf Differentiating Noise Source:}
    Given a finite set of programs $G$, a distance metric $d$, and a Loss function $\cL$,  
    a Noise Source $\rho_N$ is
    {\bf differentiating}, if
for all $\delta > 0$ and $\gamma > 0$, there exists a natural number
$k$, and $\epsilon \in \mathbb{R}^+$, such that for all $n \geq k$, for all
vectors $\vec{z}_h$ of length $n$, the following is true:
\begin{equation*}
    \begin{split}
        \rho_N[\forall \vec{z} \in Z^n.~\cL(\vec{z}, \vec{y}) - \cL(\vec{z}_h, \vec{y})
        \leq
    \gamma \\ \implies d(\vec{z}, \vec{z}_h) < \epsilon~\vert~z_h] \geq (1-\delta)
    \end{split}
\end{equation*}
\end{definition}
\Comment{
\noindent    The above condition is equivalent to
\begin{equation*}
    \begin{split}
    \rho_N\Big[
        \forall \vec{z} \in Z^n.~d(\vec{z}, \vec{z}_h) \geq \epsilon \implies \\
        \cL(\vec{z}, \vec{y}) - \cL(\vec{z}_h, \vec{y}) > \gamma
~\Big\vert~\vec{z}_h \Big] \geq (1 - \delta) 
    \end{split}
    \end{equation*}
}

If we are using the optimal Loss Function for the given Noise
Source (Subsection~\ref{subsec:optimalloss}), the above condition reduces to:
\[
    \rho_N\Big[
        \forall \vec{z} \in Z^n.~
\frac{\rho_N(\vec{y}~\vert~\vec{z}_h)}
    {\rho_N(\vec{y}~\vert~\vec{z})} \leq \gamma
    \implies
        d(\vec{z}, \vec{z}_h) < \epsilon
        ~\Big\vert~\vec{z}_h \Big] \geq (1 - \delta) 
    \]

\noindent{\bf Convergence:}\\
Convergence is guaranteed in presence of a differentiating Input
Source and a differentiating Noise Source.

\begin{theorem}
    Given a finite set of programs $G$, 
    distribution $\rho_p$ from which the programs are sampled, 
    a Loss Function $\cL$, 
    a {\bf differentiating Input Source} $\rho_i$, and  
    a {\bf differentiating Noise Source} $\rho_N$, then our synthesis
    algorithm will guarantee convergence.
\end{theorem}
\noindent We present the proof of this theorem in the
Appendix~\ref{subsec:convergence-A}.

\Comment{
\begin{proof}
    Let $\delta > 0$, $\delta_i > 0$, and $\delta_N > 0$ be two real numbers such that
    $\delta = \delta_i\delta_N$. 
    For all positive probability programs $p_h \in G$:
    \begin{equation*}
        \begin{split}
        Pr[p^n_s \approx p_h~\vert~p_h, N]  \geq 
    \int\limits_{\vec{x} \in X^n, 
    \vec{y} \in Y^n}   
            \mathds{1}\big(\forall p \in \tilde{G}_{p_h}.\\
            \cL(p[\vec{x}], \vec{y}) - \log \pi( G_{\vec{x}, p[\vec{x}]}) >
            \cL(p_h[\vec{x}], \vec{y})  - \log \pi (G_{\vec{x},
            p_h[\vec{x}]})\big)\\ 
            \rho_i(d\vec{x})\rho_N(d\vec{y}~\vert~p_h[\vec{x}])
        \end{split}
    \end{equation*}
    Let $\gamma = \argmin_{p \in G} \rho_p(G_p)$. Note that 
    \[\log \pi( G_{\vec{x}, p[\vec{x}]})   - \log \pi (G_{\vec{x},
            p_h[\vec{x}]}) > \gamma\]
     \begin{equation*}
        \begin{split}
        Pr[p^n_s \approx p_h~\vert~p_h, N]  \geq 
    \int\limits_{\vec{x} \in X^n, 
    \vec{y} \in Y^n}   
            \mathds{1}\big(\forall p \in \tilde{G}_{p_h}.\\
            \cL(p[\vec{x}], \vec{y}) - 
            \cL(p_h[\vec{x}], \vec{y})  > \gamma \big)\\ 
            \rho_i(d\vec{x})\rho_N(d\vec{y}~\vert~p_h[\vec{x}])
        \end{split}
    \end{equation*}
      \begin{equation*}
        \begin{split}
        Pr[p^n_s \approx p_h~\vert~p_h, N]  \geq 
    \int\limits_{\vec{x} \in X^n, 
    \vec{y} \in Y^n}   
            \mathds{1}\big(\forall p \in \tilde{G}_{p_h}.\\
            \cL(p[\vec{x}], \vec{y}) - 
            \cL(p_h[\vec{x}], \vec{y})  > \gamma \big)\\ 
        \rho_N(d\vec{y}~\vert~p_h[\vec{x}])
        \mathds{1}(\forall p \in G. d(p_h[\vec{x}], p[\vec{x}]) \geq \epsilon)~\rho_i(d\vec{x})
        \end{split}
    \end{equation*}
       \begin{equation*}
        \begin{split}
        Pr[p^n_s \approx p_h~\vert~p_h, N]  \geq 
    \int\limits_{\vec{x} \in X^n, 
    \vec{y} \in Y^n}   
            \mathds{1}\big(\forall z \in Z^n. 
            d(p_h[\vec{x}], \vec{z}) \geq \epsilon\\ \implies  
            \cL(\vec{z}, \vec{y}) - 
            \cL(p_h[\vec{x}], \vec{y})  > \gamma\big)\\ 
        \rho_N(d\vec{y}~\vert~p_h[\vec{x}])
        \mathds{1}(\forall p \in G. d(p_h[\vec{x}], p[\vec{x}]) \geq \epsilon)~\rho_i(d\vec{x})
        \end{split}
    \end{equation*}
       If $n \geq n_N$, $\epsilon \geq \epsilon_N$, $\delta_N > 0$,
       the following statement is true, for any $z_h$ and $\gamma = \log \pi
       (\tilde{G}_{p_h}) - \log \pi(G_{p_h})$:
\begin{equation*}
    \begin{split}
    \rho_N\Big[
        \forall \vec{z} \in Z^n.~d(\vec{z}, \vec{z}_h) \geq \epsilon \implies \\
        \cL(\vec{z}, \vec{y}) - \cL(\vec{z}_h, \vec{y}) > \gamma
~\Big\vert~\vec{z}_h \Big] \geq (1 - \delta) 
    \end{split}
    \end{equation*}
       And for $n \geq n_i$, and $\delta_i > 0$, the following is true:
\[
    \int\limits_{\vec{x} \in X^n} \mathds{1}(
    \forall p \in \tilde{G}_{p_h}.~d(p_h[\vec{x}], p[\vec{x}]) 
    > \epsilon)\rho_i(d \vec{x}) \geq (1 - \delta)
\]
       Then for $n \geq \max(n_N, n_i)$, the following is true:
       \[
        Pr[p^n_s \approx p_h~\vert~p_h, N] 
            \geq (1 - \delta_N)(1 - \delta_i) \geq (1 - \delta)
        \]
\end{proof}

}

\section{Case Studies}
\label{sec:casestudies}
In the previous section,
we proved how having a differentiating Input
Source and a differentiating Noise Source are sufficient
for the synthesis algorithm to have convergence
guarantees. Within this section, we 
will prove that the Noise Sources and Loss Functions studied in~\cite{handa2020inductive} 
fulfill these requirements, and therefore, allow the synthesis algorithm to
provide convergence guarantees.
We will then show how breaking these requirements makes convergence
impossible. We will also show the importance of picking an appropriate distance
metric $d$ which connects the Input Source to the Noise Source.

\subsection{Differentiating Input Distributions}
In the special case, where each element of the input vector $\vec{x}$ are i.i.d.,
the Differentiating Input Distribution condition can be simplified.

For any vector $\vec{x} = \tup{x_1, x_2, \ldots x_n}$, 
\[\rho_i(\vec{x}) =
\prod\limits_{j = 1}^n \rho_I(x_j)\]
Given any two equal length vectors $\vec{z} = \tup{z_1, \ldots z_n}$
and $\vec{z}_h = \tup{z'_1, \ldots z'_n}$,
let $d_i$ be a distance metric such that 
\[d(\vec{z}, \vec{z}_h) =
\sum\limits_{j=1}^n d_i(z_j, z'_j)\]
We say the input distribution $\rho_I$ is {\bf differentiating}, if for all
program $p_h \in G$, and $p \in G^C_{p_h}$ there exists
an input $x_p$, such that $\rho_I(x_p) > 0$, and
\[
    d_i(p(x_p), p_h(x_p)) > 0
\]

\begin{theorem}
    If $\rho_I$ is differentiating then $\rho_i$ is differentiating. 
\end{theorem}
\noindent We present the proof of this theorem in the Appendix (
Theorem~\ref{thm:differentiatinginput-A}).

\subsection{Differentiating Noise Distributions and optimal Loss Function}
\noindent{\bf Case 1:} We first consider the case where the noise distribution
never introduces any corruptions in the correct output.
Formally, for all $n \in \mathbb{N}$, for all $\vec{z} \in Z^n$, $
    \rho_N(\vec{z}~\vert~\vec{z}) = 1$. 
    Consider the Loss Function $\cL_{0/\infty}$ and
    the counting distance metric $d_c$:    
\begin{definition}\label{def:loss0infty}
{\bf $0/\infty$ Loss Function:} The $0/\infty$ Loss Function \\ 
    $\cL_{0/\infty}(\vec{z}, \vec{y})$ is 0 if $\vec{z} = \vec{y}$
    and $\infty$ otherwise.
\end{definition}

\begin{definition}
    \label{def:countingdistance}
    {\bf Counting Distance} The counting distance metric $d_c$ 
    counts the number of positions two equal length
    vectors disagree on, i.e.,
    \[
        d_c(\tup{z_1, \ldots z_n}, \tup{z'_1, \ldots z'_n}) =
        \sum\limits_{i=1}^n \mathds{1}(z_i \neq z'_i)
    \]
\end{definition}

    Note that for all $\gamma > 0$, for all $n \geq 1$, for all $\epsilon \geq 1$,
and for all $\vec{z}_h \in Z^n$,
    the
    following is true:
\begin{equation*}
    \begin{split}
    \int\limits_{\vec{y} \in Y^n}\mathds{1}(\forall~z\in Z^n.
    \cL(\vec{z}, \vec{y}) - \cL(\vec{z}_h, \vec{y}) \leq \gamma \implies
   \\ d(\vec{z}, \vec{z}_h) < \epsilon)\rho_N(\vec{y}~\vert~\vec{z}_h)
   \\
 = \mathds{1}(\forall~z\in Z^n.
        \cL(\vec{z}, \vec{z}_h) - \cL(\vec{z}_h, \vec{z}_h) \leq \gamma \implies
   d(\vec{z}, \vec{z}_h) < \epsilon)
   = 1
    \end{split}
    \end{equation*}
Therefore, in this case $\rho_N$ is {differentiating}.

\noindent{\bf Case 2:}
Consider the following $n$-Substitution Noise Source and $n$-Substitution Loss
Function studied in previous work~\cite{handa2020inductive}.

\begin{definition}{\bf $n$-Substitution Noise Source:}\\ The $n$-Substitution
    Noise Source $N_{nS}$, given an output vector $\tup{z_1, \ldots z_n}$
    corrupts each string $z_i$ independently. For each string
    $z = c_1 \cdots c_k$, it replaces character $c_i$ with a random character
    not equal to $c_i$ with probability $\delta_{nS}$.
\end{definition}

\begin{definition}
{\bf $n$-Substitution Loss Function:} \\ The $n$-Substitution Loss Function 
    $\cL_{nS}(\vec{z}, \vec{y})$ uses
    per-example Loss Function $L_{nS}$ that captures a  weighted sum of
    positions where the noisy output string agrees and disagrees with the output
    from the synthesized program.  If the synthesized
program produces an output that is longer or shorter than the output in the 
noisy data set, the Loss Function is $\infty$:
\[
    \cL_{nS}(\tup{z_1, \ldots z_n}, 
    \tup{y_1, \ldots y_n}) = \sum\limits_{i=1}^n L_{nS}(z_i, y_i), \mbox{ where }
\]
\begin{equation*}
    L_{nS}(z, y) = \begin{cases}
        \infty & \vert z \vert \neq \vert y \vert \\
        \sum\limits_{i=1}^{\vert z \vert} -\log \delta_i \mbox{ if } z[i]
        \neq y[i] \mbox{ else } - \log (1 - \delta_i)
    & \vert z \vert = \vert y \vert 
    \end{cases}
\end{equation*}
\end{definition}
Note that this Loss Function is a linear transformation of the $n$-Substitution Loss
Function proposed in~\cite{handa2020inductive}.

\begin{definition}{\bf Length Distance Metric}
    Given two equal length vectors of strings, 
the length distance metric $d_l$ 
    counts the number of positions, which have unequal length strings in the two
    vectors, i.e.,
    \[
        d_l(\tup{z_1, \ldots z_n}, \tup{z'_1, \ldots z'_n}) =
        \sum\limits_{i=1}^n \mathds{1}(|z_i| \neq |z'_i|)
    \]
\end{definition}

\begin{theorem}
    $n$-Substitution Loss Function $\cL_{nS}$ is the optimal Loss Function for
    the $n$-Substitution Noise Source $N_{nS}$. 
    Also, given Length Distance Metric $d_l$
    and the
    $n$-Substitution Loss Function $\cL_{nS}$, $n$-Substitution Noise Source
    $N_{nS}$ is differentiating.
\end{theorem}
\noindent We present the proof of this theorem in the Appendix (
Theorem~\ref{thm:differentiatingnsubs-A}).

Therefore, given a differentiating Input Source, the synthesis algorithm will
guarantee convergence in presence of $n$-Substitution Noise Source and
$n$-Substitution Loss Function.

\noindent{\bf Case 3:}
Consider the 1-Delete Loss Function $\cL_{1D}$ and the 1-Delete Noise Source
$N_{1D}$ studied in previous work~\cite{handa2020inductive}.

\begin{definition} {\bf 1-Delete Noise Source:} The 1-Delete Noise source
    $N_{1D}$ given an output vector $\tup{z_1, z_2, \ldots z_n}$ independently
    corrupts each output $z_i$ by deleting a random character, with probability
    $0 < \delta_{1D} < 1$. Formally,
    \[
        \rho_{N_{1D}}(\tup{y_1, \ldots y_n}~\vert~\tup{z_1, \ldots z_n}) =
        \prod_{i=1}^n\rho_{n_{1D}}(y_i~\vert~z_i)
    \]
    where
    \[
        \rho_{n_{1D}}(z~\vert~z) = 1 - \delta_{1D}
    \] \[
    \rho_{n_{1D}}(a\cdot b~\vert a \cdot x \cdot c) = \frac{1}{\mathsf{len}(a
    \cdot c \cdot b)}\delta_{1D}\]
where $c$ is a character.
\end{definition}

\begin{definition} {\bf 1-Delete Loss Function:} The 1-Delete Loss\\ Function 
    $\cL_{1D}(\tup{z_1, \ldots z_n}, \tup{y_1, \ldots y_n})$
    assigns loss 
    $-\log (1 - \delta_i)$ if the output $z_i$ from the synthesized program and the data set
    $y_i$ match exactly, $-\log \delta_i$ if a single
deletion enables the output from the synthesized program to match
    the output from the data set, and $\infty$ otherwise (for $0 < \delta_i <
    1$): 
\[
    \cL_{1D} (\tup{z_1, \ldots z_n}, \tup{y_1, \ldots y_n}) 
    = \sum\limits_{i=1}^n L_{1D}(z_i, y_i), \mbox{ where }
\]
\begin{equation*}
    L_{1D}(z, y) = \begin{cases}
        - \log (1 - \delta_i) & z = y\\
        - \log \delta_i & a \cdot x \cdot b = z \wedge a \cdot b = y \wedge \vert x \vert =
        1\\ 
        \infty & \text{ otherwise}
    \end{cases}
\end{equation*}
\end{definition}
Note that this Loss Function is a linear transformation of the 1-Delete Loss
Function proposed in~\cite{handa2020inductive}.

\begin{definition}
    {\bf DL-k Distance} The DL-k Distance metric $d_{DLk}$, given two vectors of
    strings, returns the count of string pairs with Damerau-Levenshtein distance 
    greater than equal to $k$ between them. Formally,
    \[
        d_{DL}(\tup{z_1, \ldots z_n}, \tup{z'_1, \ldots z'_n}) =
        \sum\limits_{i=1}^n \mathds{1}(L_{z_i, z'_i}(|z_i|, |z'_i|) \geq k)
    \]
where, $L_{a, b}(i, j)$ is the {\it Damerau-Levenshtein} metric~\cite{damerau1964technique}.
\end{definition}

\begin{theorem}
    1-Delete Loss Function $\cL_{1D}$ is the optimal Loss Function for
    the 1-Delete Noise Source $N_{1D}$. 
    Also, given DL-2 Distance Metric $d_{DL2}$~\cite{damerau1964technique}
    and the
    1-Delete Loss Function $\cL_{1D}$, 1-Delete Noise Source
    $N_{1D}$ is differentiating.
\end{theorem}
\noindent We present the proof of this theorem in the 
Appendix ( Theorem~\ref{thm:differentiating1D-A})

Therefore, given a differentiating Input Source, the synthesis algorithm will
guarantee convergence in presence of 1-Delete Noise Source and
1-Delete Loss Function.

\noindent{\bf Case 4:}
Consider the Damerau-Levenshtein Loss Function $\cL_{DL}$ and the 1-Delete Noise Source
$N_{1D}$ studied in previous work~\cite{handa2020inductive}.

\begin{definition}
{\bf Damerau-Levenshtein (DL) Loss Function:} The DL Loss Function 
$\cL_{DL}(p, \cD)$ uses the {\it Damerau-Levenshtein} metric~\cite{damerau1964technique},
to measure the distance between the output from the synthesized program and the
corresponding output in the noisy data set: 
\[
    \cL_{DL}(p, \cD) = \sum\limits_{(x_i, y_i) \in \cD} 
    L_{p(x_i), y_i}\big(\left\vert p(x_i)
    \right\vert, \left\vert y_i \right\vert\big)
\]
where, $L_{a, b}(i, j)$ is the 
    {\it Damerau-Levenshtein} metric~\cite{damerau1964technique}.
\end{definition}
This metric counts the number of 
single character deletions, insertions, substitutions, or transpositions
required to convert one text string into another. 
Because more than $80\%$ of all human misspellings 
are reported to be captured by a single one of these four 
operations~\cite{damerau1964technique}, the DL Loss Function may
be appropriate for computations that work with human-provided
text input-output examples. 

\begin{theorem}
   Given DL-2 Distance Metric $d_{DL2}$
    and the
    Damerau-Levenshtein Loss Function $\cL_{DL}$, the 1-Delete Noise Source
    $N_{1D}$ is differentiating.
\end{theorem}
\noindent We present the proof of this theorem in the Appendix (
Theorem~\ref{thm:differentiating1DDL-A}).

Therefore, given a differentiating Input Source, the synthesis algorithm will
guarantee convergence in presence of 1-Delete Noise Source and
Damerau-Levenshtein Loss Function.

\noindent{\bf Case 5:}
Consider the Damerau-Levenshtein Loss Function $\cL_{DL}$ and the 
$n$-Substitution Noise Source
$N_{nS}$ studied in previous work~\cite{handa2020inductive}.

\begin{theorem}
   Given Length Distance Metric $d_l$
    and the
    Damerau-Levenshtein Loss Function $\cL_{DL}$, the $n$-Substitution Noise Source
    $N_{nS}$ is differentiating.
\end{theorem}
\noindent We present the proof of this theorem in the Appendix (
Theorem~\ref{thm:differentiatingnSDL-A}).

Therefore, given a differentiating Input Source, the synthesis algorithm will
guarantee convergence in presence of $n$-Substitution Noise Source and
Damerau-Levenshtein Loss Function.

\noindent{\bf Connections to Empirical Results:}
Handa et al. within their paper~\cite{handa2020inductive}, 
empirically evaluated Damerau-Levenshtein
and 1-Delete Loss Functions in presence of 1-Delete style noise.
Both 1-Delete Loss Function and Damerau-Levenshtein Loss Function are able to
synthesize the correct program in presence of some noise. This is in line with
our convergence results. 1-Delete Loss Function is also able to tolerate 
datasets with more noise, compared to Damerau-Levenshtein Loss Function.
Even when all input-output examples were corrupted, 
1-Delete Loss Function was able to synthesize the correct program. 
Note that 1-Delete Loss Function is the optimal Loss Function in presence of
1-Delete Noise Source, therefore it has a higher probability of synthesizing the
correct program.

They also evaluated $n$-Substitution Loss Function in presence of the
$n$-Substitution Noise Source. 
Inline with our theoretical results, their technique was able to synthesize the
correct answer over datasets corrupted by $n$-Substitution Noise Source.

\subsection{Non-Differentiating Input Distributions}
Let $G$ be a set of two programs $p_1$ and $p_2$ and let
$x^*$ be the only input on which $p_1$ and $p_2$ disagree on, i.e., $p_1(x^*) \neq
p_2(x^*)$ and $\forall x \in X, x \neq x^* \implies p_1[x] = p_2[x]$.
Let $\rho_p$ be a probability distribution over programs in $G$. 
Without loss of generality, 
we assume
$C(p_2) \geq C(p_1)$. Let $\rho_i$ be the probability distribution over
input vectors.

\noindent{\bf Case 1:} The input process 
never returns a differentiating input, i.e.,
for all $n \in \mathbb{N}$ and for all $\vec{x} \in X^n$,
$x^* \in \vec{x} \implies \rho_i(\vec{x}) = 0$.  

For all $n \in \mathbb{N}$ and for all $\epsilon > 0$, the following statement
is true for both $p_h = p_1$ and $p_2$:
\begin{equation*}
    \begin{split}
    \int\limits_{\vec{x} \in X^n} \mathds{1}(\forall p \in \tilde{G}_{p_h}.~
    d(p[\vec{x}], p_h[\vec{x}]) > \epsilon) \rho_i(d\vec{x}) 
\\
        \leq 
        \int\limits_{\vec{x} \in X^n} \mathds{1}(p_1[\vec{x}] \neq p_2[\vec{x}])
        \rho_i(d\vec{x})
        = 0
    \end{split}
    \end{equation*}
Therefore, $\rho_i$ is non-differentiating.

\begin{theorem}
In this case, $Pr[p^n_s \approx p_1~\vert~p_1, N] = 0$ 
\end{theorem}
\noindent We present the proof of this theorem in the Appendix (
Theorem~\ref{thm:nondifferentiatinginputcase1-A}).

Since $p_2[\vec{x}] = p_1[\vec{x}]$ for all $\vec{x}$, the best move for any
algorithm is to always predict the simplest program (which in this case is
$p_2$). Hence, if the hidden program is not the simplest program, even in
the absence of noise, the synthesis algorithm will not converge to the correct
hidden program.

\noindent{\bf Case 2:}
The input process 
never returns a differentiating input with sufficient probability, i.e.,
there exists a $\delta_i$, such that 
for all $n \in \mathbb{N}$,
\[
\int\limits_{\vec{x} \in X^n} \mathds{1}(x^* \in \vec{x})\rho_i(d\vec{x}) <
\delta_i
\]
For all $n \in \mathbb{N}$ and for all $\epsilon > 0$, the following statement
is true for both $p_h = p_1$ or $p_2$:
\begin{equation*}
    \begin{split}
    \int\limits_{\vec{x} \in X^n} \mathds{1}(\forall p \in \tilde{G}_{p_h}.~
    d(p[\vec{x}], p_h[\vec{x}]) > \epsilon) \rho_i(d\vec{x}) 
    \\
        \leq
    \int\limits_{\vec{x} \in X^n} \mathds{1}(
    p_1[\vec{x}] \neq p_2[\vec{x}]) \rho_i(d\vec{x})  
< \delta_i
    \end{split}
\end{equation*}
Therefore, $\rho_i$ is non-differentiating. 

\begin{theorem}
In this case, 
$Pr[p^n_s \approx p_1~\vert~p_1, N] < \delta_i$. 
\end{theorem}
\noindent We present the proof of this theorem in the Appendix (
Theorem~\ref{thm:nondifferentiatinginputcase2-A}).

Hence, no matter how much we increase $n$ (i.e., the size of the dataset
sampled), the probability that $p_1$ will be synthesized (in presence of any
Loss Function) is less than $\delta_i$. 

\subsection{Non-Differentiating Noise Distributions}
\label{subsec:nondifferentiatingnoise}
\noindent{\bf Case 1:}
We first consider the case where the Noise Source corrupts all information
identifying the hidden program. 
Let $G$ be a set containing two programs, $p_a$ and $p_b$.
$p_a$ takes an input string $x$ and appends character $``a"$ in front of it. $p_b$, 
similarly, takes an input string $x$ and appends character $``b"$ in front of it.
\[
\begin{array}{rcl}
    p_a &:=& \mathsf{append}(``a", x)
\\
    p_b &:=& \mathsf{append}(``b", x)
\end{array}
\]
\Comment{
Let $\rho_p$ be a probability distribution over programs in $G$. 
Without loss of generality, 
we assume
$C(p_b) \geq C(p_a)$. Let $\rho_i$ be the probability distribution over
input vectors.
}

Let $N$ be a Noise Source which deletes the first
character of the output string with probability $1$. Note that in presence of
this Noise Source, no synthesis algorithm can infer which of the given programs
$p_a$ or $p_b$ is the hidden program.

\begin{theorem}
$
    Pr[p^n_s \approx p_a~\vert~p_a, N] + Pr[p^n_s \approx p_b~\vert~p_b, N] \leq
    1
$
\end{theorem}
\noindent We present the proof of this theorem in the Appendix (
Theorem~\ref{thm:nondifferentiatingnoisecase1-A}).

Therefore, any gains in improving convergence of the synthesis algorithm by
increasing $n$, assuming $p_a$ is the hidden program, will be on the cost of
convergence when $p_b$ is the hidden program. Formally, 
\[Pr[p^n_s \approx
p_a~\vert~p_a, N] \geq 1 - \delta \implies Pr[p^n_s \approx p_b~\vert~p_b, N]
\leq \delta\]

\noindent{\bf Case 2:}
The choice of Loss Function affects whether the Noise Source is differentiating or
not.
For example, consider a Noise Source which reveals 
information identifying the hidden program
with high probability, but the Loss Function does not capture this information.
1-Delete Noise Source $N_{1D}$, given the 1-Delete Loss
function $\cL_{1D}$ (and DL-2 distance metric $d_{DL2}$), is 
differentiating.

Since, $n$-Substitution Loss Function $\cL_{nS}$ penalizes a deletion with infinite
loss, the following is true:
\begin{theorem}
 Given $n$-Substitution Loss Function
$\cL_{nS}$, 1-Delete Noise Source $N_{1D}$ is non-differentiating.
In this case, the synthesis algorithm will never converge.
\end{theorem}
\noindent We present the proof of this theorem in the Appendix (
Theorem~\ref{thm:nondifferentiatingnoisecase2-A}).

Handa et al. empirically showed within their paper~\cite{handa2020inductive}
that the $n$-Substitution Noise Source requires all input-output examples to be
correct (i.e., it cannot tolerate any noise) when the specific noise is
introduced by the 1-Delete Noise Source. This theorem is in line with their
experimental results.

\subsection{Necessity of the distance metric $d$}
The distance metric serves as an important link between the Input Source and the
Noise Source. 
A simple distance metric (like the counting distance) 
makes it easy to  construct a differentiating Input Source but a simple distance metric
lack the restrictions required to prove a Noise Source differentiating.
Similarly, a restricted distance metric may enable us to easily prove the Noise
Source differentiating but make it hard to construct a differentiating Input
Source.

For example, consider the following case. 
Let $G$ be a set containing two programs $p_a$ and $p_b$.
$p_a$ takes an input tuple (string and a boolean) $(x, b)$ and appends character 
$``a"$ in front of $x$ if $b$ is true, else it appends $``aa"$ in front of $x$. 
$p_b$, 
similarly, takes an input tuple (string and a boolean) $(x, b)$ and 
appends character $``b"$ in front of $x$ if $b$ is true, else it appends $``bb"$
in front of $x$. Formally,
\[
\begin{array}{rcl}
    p_a &:=& \mathsf{append}(\ife{b}{``a"}{``aa"}, x)
\\
    p_b &:=& 
    \mathsf{append}(\ife{b}{``b"}{``bb"}, x)
\end{array}
\]
\Comment{
Let $\rho_p$ be a probability distribution over programs in $G$. 
Without loss of generality, 
we assume
$C(p_b) \geq C(p_a)$. Let $\rho_i$ be the probability distribution over input
vectors.
}
Consider a Noise Source which deletes the first character with probability $1$.
For counting distance metric $d_c$,
an Input Source $\rho_i$ which only
returns inputs with $b$ set to true, is differentiating, i.e., 
the outputs produced by $p_a$ and $p_b$ will be of the from $c \cdot x$, where
$c$ is $``a"$ and $``b"$ respectively.

\begin{theorem}
There exists no Loss Functions for which the above Noise Source is
differentiating. Note that for $\rho_i$ described above, the synthesis algorithm
will not converge as well.
\end{theorem}
\noindent We present the proof of this theorem in the Appendix (
Theorem~\ref{thm:necessitynondifferentiating-A}).

But for DL-2 Distance metric $d_{DL2}$ (which only counts
the number of disagreements have atleast 2 Damerau-Levenshtein between them),
then an Input Source $\rho_i$ has to return inputs with $b$ set to true, with
high probability, to be differentiating.

\begin{theorem}
    Consider a Loss Function $\cL_{ab}$ 
    which checks if the first character appended to a
    string is either 
$``a"$ or
$``b"$.
    Given DL-2 Distance metric $d_{DL2}$ and the Loss Function $\cL_{ab}$,
    the Noise Source described above is differentiating.

In this case,
if we pick a differentiating Input Source, our synthesis algorithm
will have convergence guarantees.
\end{theorem}
\noindent We present the proof of this theorem in the Appendix (
Theorem~\ref{thm:necessitydifferentiating-A}).

But note that for space of programs considered in
Subsection~\ref{subsec:nondifferentiatingnoise} (case 1), even through the Noise
Source is differentiating, the complex 
$d_{DL2}$ distance metric will not work as a differentiating
Input Source is impossible to construct for $d_{DL2}$ distance.

\section{Related Work}

\noindent{\bf Noise-Free Programming-by-examples:}
The problem of learning programs from a set of correct input-output 
examples has been studied extensively~\cite{shaw1975inferring, 
gulwani2011automating, singh2016transforming}.
Even though these techniques provide correctness guarantees (and convergence
guarantees if all inputs which allow us to differentiate it from other programs
are provided), these techniques cannot handle noisy datasets.

\noindent{\bf Noisy Program Synthesis:}
There has been recent interest in synthesizing programs over noisy
datasets~\cite{handa2020inductive, peleg2020perfect}. 
These systems do not formalize the requirements for their algorithms to converge to the
correct underlying program. Handa el al's work uses Loss Functions as a
parameter in their synthesis procedure to tailor it to different noise sources. 
It does not provide any process to either pick or design loss functions, given
information about the noise source. 

\noindent{\bf Neural Program Synthesis:}
There is extensive work that uses 
machine learning/deep neural networks to synthesize programs~\cite{raychev2016learning, 
devlin2017robustfill, balog2016deepcoder}. These techniques require a training
phase, a differentable loss function, and provide no formal 
correctness or convergence guarantees.
There also has been a lack of work on designing an appropriate loss
function, given information about the noise source.

\noindent{\bf Learning Theory:}
Learning theory captures the formal aspects of learning models over noise 
data~\cite{illeris2018overview, bolles1975learning, kearns1994introduction}.
Our work takes concepts from learning theory and applies them to the specific
context of synthesizing programs over noisy data. To the best of our knowledge,
the special case of noisy program synthesis has never been explored in learning
theory.

There has considerable work done on designing loss functions for training neural
networks~ 
\cite{xu2019l_dmi, ghosh2017robust}. To the best of our knowledge, these works
do not theoretically prove the optimality of their loss functions with respect
to a given noise source.

\section{Conclusion}
Learning models from noisy data with guarantees is an important problem. 
There has been recent work on synthesizing programs over noisy
datasets using loss functions. Even though these systems have delivered
effective program synthesis over noisy datasets, they do not provide any guidance
for constructing loss functions given prior information about the noise source, nor
do they comment on the convergence of their algorithms to the correct underlying
program.

We are the first paper to formalize the hidden process which samples a hidden
program and constructs a noisy dataset, thus formally specifying the correct
solution to a given noisy synthesis problem.
We are the first paper to formalize the concept of an optimal loss function for
noisy synthesis and provide a closed form expression for such a loss function,
in presence of perfect and imperfect information about the noise source.
We are the first paper of formalize the constraints on the input source, the
noise source, and the loss function which allow our synthesis algorithm to
eventually converge to the correct underlying program.
The case studies highlight why these constraints are necessary.
We also provide proofs of convergence for the noisy program synthesis problems
studied in literature.

\Comment{
\begin{acks}                            
  This material is based upon work supported by the
  \grantsponsor{GS100000001}{National Science
    Foundation}{http://dx.doi.org/10.13039/100000001} under Grant
  No.~\grantnum{GS100000001}{nnnnnnn} and Grant
  No.~\grantnum{GS100000001}{mmmmmmm}.  Any opinions, findings, and
  conclusions or recommendations expressed in this material are those
  of the author and do not necessarily reflect the views of the
  National Science Foundation.
\end{acks}
}


\bibliography{citation}
\balance
\appendix
\onecolumn
\section{Appendix}
\subsection{Convergence}
\label{subsec:convergence-A}
\begin{theorem}\label{thm:convergence-A}
    Given a finite set of programs $G$, a prior 
    distribution $\rho_p$, 
    a Loss Function $\cL$, 
    a {\bf differentiating Input Source} $\rho_i$, and  
    a {\bf differentiating Noise Source} $\rho_N$, the synthesis
    algorithm~\ref{alg:regsyn} will have convergence guarantees.
\end{theorem}
\begin{proof}
    Given a $\delta > 0$, let $\delta_i > 0$ and $\delta_N > 0$ be two real
    numbers such that
    $\delta = \delta_i + \delta_N$. 
    
    For all positive probability programs $p_h \in G$:
    \begin{equation*}
        \begin{split}  
  Pr[p^n_s \approx p_h~\vert~p_h, N] 
            \geq 
    \int\limits_{\vec{x} \in X^n, 
    \vec{y} \in Y^n}   
            \mathds{1}\big(\forall p \in G^C_{p_h}.
            \cL(p[\vec{x}], \vec{y}) - \log \rho_P( G_{\vec{x}, p[\vec{x}]}) >
            \cL(p_h[\vec{x}], \vec{y})  - \log \pi (G_{\vec{x},
            p_h[\vec{x}]})\big) 
            \\ 
            \rho_i(d\vec{x})\rho_N(d\vec{y}~\vert~p_h[\vec{x}])
        \end{split}
    \end{equation*}
    Let $\gamma = \min_{p \in G} \rho_p(G_p)$. Note that 
    $\log \pi( G_{\vec{x}, p[\vec{x}]})   - \log \pi (G_{\vec{x},
            p_h[\vec{x}]}) > \gamma$
     \begin{equation*}
        \begin{split}
        Pr[p^n_s \approx p_h~\vert~p_h, N]  \geq 
    \int\limits_{\vec{x} \in X^n, 
    \vec{y} \in Y^n}   
            \mathds{1}\big(\forall p \in G^C_{p_h}.
            \cL(p[\vec{x}], \vec{y}) - 
            \cL(p_h[\vec{x}], \vec{y})  > \gamma \big)
            \rho_i(d\vec{x})\rho_N(d\vec{y}~\vert~p_h[\vec{x}])
        \end{split}
    \end{equation*}
      \begin{equation*}
        \begin{split}
        \geq 
    \int\limits_{\vec{x} \in X^n, 
    \vec{y} \in Y^n}   
            \mathds{1}\big(\forall p \in G^C_{p_h}.
            \cL(p[\vec{x}], \vec{y}) - 
            \cL(p_h[\vec{x}], \vec{y})  > \gamma \big) 
        \rho_N(d\vec{y}~\vert~p_h[\vec{x}])
            \\ \mathds{1}(\forall p \in G^C_{p_h}. d(p_h[\vec{x}], p[\vec{x}]) \geq \epsilon)~\rho_i(d\vec{x})
        \end{split}
    \end{equation*}
    If $\forall p \in G^C_{p_h}. d(p_h[\vec{x}], p[\vec{x}]) \geq \epsilon$,
$\big(\forall \vec{z} \in Z^n. 
            d(p_h[\vec{x}], \vec{z}) \geq \epsilon \implies  
            \cL(\vec{z}, \vec{y}) - 
            \cL(p_h[\vec{x}], \vec{y})  > \gamma\big) $
         $\implies$
$\big(\forall p \in G^C_{p_h}.
            \cL(p[\vec{x}], \vec{y}) - 
            \cL(p_h[\vec{x}], \vec{y})  > \gamma \big) $.
            Therefore,
    \begin{equation*}
        \begin{split}
         \geq 
    \int\limits_{\vec{x} \in X^n, 
    \vec{y} \in Y^n}   
            \mathds{1}\big(\forall \vec{z} \in Z^n. 
            d(p_h[\vec{x}], \vec{z}) \geq \epsilon \implies  
            \cL(\vec{z}, \vec{y}) - 
            \cL(p_h[\vec{x}], \vec{y})  > \gamma\big) 
       \\ \rho_N(d\vec{y}~\vert~p_h[\vec{x}])
        \mathds{1}(\forall p \in G. d(p_h[\vec{x}], p[\vec{x}]) \geq \epsilon)~\rho_i(d\vec{x})
        \end{split}
    \end{equation*}
     Since $\rho_N$ is differentiating, we can assume that that for $\delta_N >
     0$, 
     If $n \geq n_N$, $\epsilon \geq \epsilon_N$, $\delta_N > 0$,
       the following statement is true, for any $z_h$ and $\gamma = \log \pi
       (G^C_{p_h}) - \log \pi(G_{p_h})$:
\begin{equation*}
    \begin{split}
    \rho_N\Big[
        \forall \vec{z} \in Z^n.~d(\vec{z}, \vec{z}_h) \geq \epsilon \implies 
        \cL(\vec{z}, \vec{y}) - \cL(\vec{z}_h, \vec{y}) > \gamma
~\Big\vert~\vec{z}_h \Big] \geq (1 - \delta_N) 
    \end{split}
    \end{equation*}
    Since $\rho_i$ is differentiating, we can assume that for
    $\delta_i > 0$ there exists an $n_i$, 
       for $n \geq n_i$, the following is true:
\[
    \int\limits_{\vec{x} \in X^n} \mathds{1}(
    \forall p \in G^C_{p_h}.~d(p_h[\vec{x}], p[\vec{x}]) 
    > \epsilon)\rho_i(d \vec{x}) \geq (1 - \delta_i)
\]
       Then for $n \geq \max(n_N, n_i)$, the following is true:
       \[
        Pr[p^n_s \approx p_h~\vert~p_h, N]             
\geq (1 - \delta_N)
    \int\limits_{\vec{x} \in X^n} \mathds{1}(
    \forall p \in G^C_{p_h}.~d(p_h[\vec{x}], p[\vec{x}]) 
    > \epsilon)\rho_i(d \vec{x}) 
\geq (1 - \delta_N)(1 - \delta_i) \geq (1 - \delta)
        \]
\end{proof}

\subsection{Supplementary Material: Case Studies}

\noindent{\bf Differentiating Input Distributions}
\begin{theorem}\label{thm:differentiatinginput-A}
    If $\rho_I$ is differentiating then $\rho_i$ is differentiating. 
\end{theorem}
\begin{proof}
    Let for all $p \in G^C_{p_h}$,
    $x_p$ be an input such that $\rho_I(x_p) > 0$
    \[
        d_i(p(x_p), p_h(x_p))> 0
    \]
    Let $\delta_i$ and $\epsilon_i$ be the largest rational numbers such that,
    for all $p$, $\rho_I(x_p) \geq \delta_i$ and $d_i(p(x_p), p_h(x_p)) \geq
    \epsilon_i$.

    Let $\epsilon$ and $\delta$ be any rational number greater than $0$.
    Let $m$ and $n_0$ be natural numbers such that $m \geq
    \frac{\epsilon}{\epsilon_i}$, for $n \geq |G|n_0$,
    \[
        \sum\limits_{j=0}^m C^{n_0}_j \delta_i^j(1-\delta_i)^{{n_0}-j} \leq
        \frac{\delta}{|G|}
    \]
 \[
    \int\limits_{\vec{x} \in X^n} \mathds{1}(
    \forall p \in G^C_{p_h}.~d(p_h[\vec{x}], p[\vec{x}]) 
    > \epsilon)\rho_i(d \vec{x})
\]
  \[
  \geq  \int\limits_{\vec{x} \in X^n} \mathds{1}(
    \forall p \in G^C_{p_h}.~d(p_h[\vec{x}], p[\vec{x}]) 
    > m\epsilon_i)\rho_i(d \vec{x})
\]
    \[
        \geq \int\limits_{\vec{x} \in X^n} \mathds{1}(\forall~p \in G^C_{p_h}.
        \text{atleast }m~x_p\text{ occur in } \vec{x}) \rho_i(d \vec{x})
    \]
      \[
          = (1 - \sum\limits_{j=0}^{m} C^{n_0}_j \delta_i^j(1-\delta_i)^{n_0 -
          j})^{|G|}\geq  (1 - \frac{\delta}{|G|})^{|G|} \geq 1-\delta 
      \]
\end{proof}

\noindent{\bf Differentiating Noise Distributions and optimal Loss Function}

\noindent{\bf Case 2:}
\begin{theorem}\label{thm:differentiatingnsubs-A}
    $n$-Substitution Loss Function $\cL_{nS}$ is the optimal Loss Function for
    the $n$-Substitution Noise Source $N_{nS}$. 
    Also, given Length Distance Metric $d_{l}$
    and the
    $n$-Substitution Loss Function $\cL_{nS}$, $n$-Substitution Noise Source
    $N_{nS}$ is differentiating.
\end{theorem}
\begin{proof}
\begin{equation*}
\begin{split}
- \log \rho_{N_{nS}
}(\tup{y_1, \ldots y_n}~\vert~\tup{z_1, \ldots z_n}) = 
\sum\limits_{i=1}^n -\log \rho_{n_{nS}}(y_i~\vert~z_i)
\end{split}
\end{equation*}
where:
\begin{equation*}
\begin{split}
- \log \rho_{n_{nS}}(s'_1 \cdot \ldots \cdot s'_k~\vert~
s_1 \cdot \ldots \cdot s_k) = \sum\limits_{i=1}^n (- \mathds{1}(s'_i = s_i)\log
(1 - \delta_i)) + (- \mathds{1}(s'_i \neq s_i) \log \delta_i )
\end{split}
\end{equation*}
Note that $\cL_{nS} = - \log \rho_{N_{nS}}$. Hence $n$-Substitution Loss
Function $\cL_{nS}$ is the optimal
Loss Function for $n$-Substitution Noise Source.

    if $d_l(\vec{z}, \vec{z}_h) \geq 1$, then for all samples $\vec{y}$,
    $\cL_{nS}(\vec{z}, \vec{y}) = \infty$.
 Therefore for all $\gamma > 0$, $\delta > 0$, 
\[
\rho_N[
\forall~\vec{z} \in Z^n.d_{l}(\vec{z}, \vec{z}_h) \geq 1 \implies 
    \cL(\vec{z}, \vec{y}) - \cL(\vec{z}_h, \vec{y}) = \infty >
    \gamma~\vert~\vec{z}_h
] = 1 \geq (1- \delta)
\]

\end{proof}

\noindent{\bf Case 3:}
\begin{theorem}\label{thm:differentiating1D-A}
    1-Delete Loss Function $\cL_{1D}$ is the optimal Loss Function for
    the 1-Delete Noise Source $N_{1D}$. 
    Also, given DL-2 Distance Metric $d_{DL2}$
    and the
    1-Delete Loss Function $\cL_{1D}$, 1-Delete Noise Source
    $N_{1D}$ is differentiating.
\end{theorem}
\begin{proof}
\begin{equation*}
\begin{split}
- \log \rho_{N_{1D}
}(\tup{y_1, \ldots y_n}~\vert~\tup{z_1, \ldots z_n}) = 
\sum\limits_{i=1}^n -\log \rho_{n_{1D}}(y_i~\vert~z_i)
\end{split}
\end{equation*}
where $\rho_{n_{1D}}(z~\vert~z) = (1 - \delta_i)$ and $\rho_{n_{1D}}(y~\vert~z)
= \delta_i$ if $y$ has exactly one character deleted with respect to $z$.

Note that $\cL_{1D} = - \log \rho_{N_{1D}}$. Hence 1-Delete Loss
Function $\cL_{1D}$ is the optimal
Loss Function for 1-Delete Noise Source.

If $d_{DL2}(\vec{z}, \vec{z}_h) \geq 1$ then $\cL_{1D}(\vec{z}, \vec{y}) =
\infty$.
Therefore for all $\gamma > 0$, $\delta > 0$, 
\[
\rho_N[
\forall~\vec{z} \in Z^n.d_{DL2}(\vec{z}, \vec{z}_h) \geq 1 \implies 
\cL(\vec{z}, \vec{y}) - \cL(\vec{z}_h, \vec{y}) = \infty >
    \gamma~\vert~\vec{z}_h] = 1 \geq (1- \delta)
\]
\end{proof}

\noindent{\bf Case 4:}
\begin{theorem}\label{thm:differentiating1DDL-A}
   Given DL-2 Distance Metric $d_{DL2}$
    and the
    Damerau-Levenshtein Loss Function $\cL_{DL}$, the 1-Delete Noise Source
    $N_{1D}$ is differentiating.
\end{theorem}
\begin{proof}
If $d_{DL2}(\vec{z}, \vec{z}_h) \geq m$,
then $\cL_{DL}(\vec{z}, \vec{y}) - \cL_{DL}(\vec{z}, \vec{y}) \geq m$ .
Therefore for all $\gamma > 0$, $\delta > 0$, 
\[
\rho_N[
\forall~\vec{z} \in Z^n.d_{DL2}(\vec{z}, \vec{z}_h) \geq \gamma \implies 
\cL(\vec{z}, \vec{y}) - \cL(\vec{z}_h, \vec{y}) > \gamma~\vert~\vec{z}_h
] = 1 \geq (1- \delta)
\]
\end{proof}

\noindent{\bf Case 5:}
\begin{theorem}\label{thm:differentiatingnSDL-A}
   Given Length Distance Metric $d_{l}$ and the
    Damerau-Levenshtein Loss Function $\cL_{DL}$, the $n$-Substitution Noise Source
    $N_{nS}$ is differentiating.
\end{theorem}
\begin{proof}
If $d_{l}(\vec{z}, \vec{z}_h) \geq m$,
    then $\cL_{DL}(\vec{z}, \vec{y}) - \cL_{DL}(\vec{z}, \vec{y}) \geq m$
    (Assuming replacements take place from a very large set).
Therefore for all $\gamma > 0$, $\delta > 0$, 
\[
\rho_N[
\forall~\vec{z} \in Z^n.d_{l}(\vec{z}, \vec{z}_h) \geq \gamma \implies 
\cL(\vec{z}, \vec{y}) - \cL(\vec{z}_h, \vec{y}) > \gamma~\vert~\vec{z}_h
] = 1 \geq (1- \delta)
\]
\end{proof}

\noindent{\bf Non-Differentiating Input Distributions}

\noindent{\bf Case 1:} 
\begin{theorem}\label{thm:nondifferentiatinginputcase1-A}
In this case, $Pr[p^n_s \approx p_1~\vert~p_1, N] = 0$ 
\end{theorem}
\begin{proof}
\begin{equation*}
    \begin{split}
        Pr[p^n_s \approx p_1~\vert~p_1, N] \leq
       \int\limits_{\vec{x} \in X^n, \vec{y} \in Y^n} 
        \mathds{1}\big(
        \cL(p_2[\vec{x}], \vec{y}) - \\ 
        \log \rho_p(G_{\vec{x}, p_2[\vec{x}]}) 
        > \cL(p_1[\vec{x}], \vec{y}) - \log \rho_p(G_{\vec{x}, p_1[\vec{x}]})
        \big)
         \rho_N(d\vec{y}~\vert~p_1[\vec{x}])\rho_i(d\vec{x})
        \\
        =
       \int\limits_{\vec{x} \in X^n} 
        \mathds{1}\big(
        \cL(p_2[\vec{x}], p_1[\vec{x}]) - 
         \log \rho_p(G_{\vec{x}, p_2[\vec{x}]}) 
         > \cL(p_1[\vec{x}], p_1[\vec{x}]) - \log \rho_p(G_{\vec{x}, p_1[\vec{x}]})
        \big) \rho_i(d\vec{x})
\\ = 0
    \end{split}
\end{equation*}
\end{proof}
\noindent{\bf Case 2:}
\begin{theorem}\label{thm:nondifferentiatinginputcase2-A}
In this case, 
$Pr[p^n_s \approx p_1~\vert~p_1, N] < \delta_i$. 
\end{theorem}
\begin{proof}
\begin{equation*}
    \begin{split}
        Pr[p^n_s \approx p_1~\vert~p_1, N] \leq
        \\
       \int\limits_{\vec{x} \in X^n, \vec{y} \in Y^n} 
        \mathds{1}\big(
        \cL(p_2[\vec{x}], \vec{y}) -  
        (\log \rho_p(G_{\vec{x}, p_2[\vec{x}]})) 
        > \cL(p_1[\vec{x}], \vec{y}) - \log \rho_p(G_{\vec{x}, p_1[\vec{x}]})
        \big)
         \rho_N(d\vec{y}~\vert~p_1[\vec{x}])\rho_i(d\vec{x})
        \\
        =
       \int\limits_{\vec{x} \in X^n} 
        \mathds{1}\big(
        \cL(p_2[\vec{x}], p_1[\vec{x}]) - 
         (\log \rho_p(G_{\vec{x}, p_2[\vec{x}]})) 
        > \cL(p_1[\vec{x}], p_1[\vec{x}]) - \log \rho_p(G_{\vec{x},
        p_1[\vec{x}]})
        \big) \rho_i(d\vec{x}) \\
        \leq \int\limits_{\vec{x} \in X^n} \mathds{1}(p_1[\vec{x}] \neq
        p_2[\vec{x}])\rho_i(d\vec{x}) < \delta_i
    \end{split}
\end{equation*}
\end{proof}

\noindent{\bf Non-Differentiating Noise Distributions}\\
\noindent{\bf Case 1:}
\begin{theorem}\label{thm:nondifferentiatingnoisecase1-A}
$
    Pr[p^n_s \approx p_a~\vert~p_a, N] + Pr[p^n_s \approx p_b~\vert~p_b, N] \leq
    1
$
\end{theorem}
\begin{proof}
For distance metric $d_c$, an Input Source $\rho_i$ which only returns $\tup{x,
\mathsf{true}}$ is differentiating, as for all $x$, $p_a(\tup{x, \mathsf{true}})
\neq p_b(\tup{x, \mathsf{true}})$.
Note that for all Loss Functions $\cL$,
\begin{equation*}
\begin{split}
\rho_N[
\cL(\tup{x_1, \ldots x_n}, \tup{b \cdot x_1, \ldots b \cdot x_n}) - 
\cL(\tup{x_1, \ldots x_n}, \tup{a \cdot x_1, \ldots a \cdot x_n}) >
\gamma~\vert~\tup{a~\cdot x_1, \ldots a \cdot x_n}
] 
\\ + 
\rho_N[
\cL(\tup{x_1, \ldots x_n}, \tup{a \cdot x_1, \ldots a \cdot x_n}) - 
\cL(\tup{x_1, \ldots x_n}, \tup{b \cdot x_1, \ldots b \cdot x_n}) >
\gamma~\vert~\tup{b~\cdot x_1, \ldots b \cdot x_n}
] 
\\ = \rho_N[
\cL(\tup{x_1, \ldots x_n}, \tup{b \cdot x_1, \ldots b \cdot x_n}) - 
\cL(\tup{x_1, \ldots x_n}, \tup{a \cdot x_1, \ldots a \cdot x_n}) >
\gamma~\vert~\tup{b~\cdot x_1, \ldots b \cdot x_n}
] 
\\ + 
\rho_N[
\cL(\tup{x_1, \ldots x_n}, \tup{a \cdot x_1, \ldots a \cdot x_n}) - 
\cL(\tup{x_1, \ldots x_n}, \tup{b \cdot x_1, \ldots b \cdot x_n}) >
\gamma~\vert~\tup{b~\cdot x_1, \ldots b \cdot x_n}
] 
 \\ \leq 1
\end{split}
\end{equation*}
\end{proof}

\noindent{\bf Case 2:}
\begin{theorem}\label{thm:nondifferentiatingnoisecase2-A}
 Given $n$-Substitution Loss Function
$\cL_{nS}$, 1-Delete Noise Source $N_{1D}$ is non-differentiating.
In this case, the synthesis algorithm will never converge.
\end{theorem}
\begin{proof}
If even a single deletions happen in $\vec{y}$, then $\cL_{nS}(\vec{z}_h,
\vec{y}) = \infty$.
The prob of no deletions is equal to $(1 - \delta_i)^n$ which decreases with
$n$, therefore in this case 1-Delete Noise Source is non-differentiating.
\end{proof}

\noindent{\bf Necessity of the distance metric $d$}
\begin{theorem}\label{thm:necessitynondifferentiating-A}
There exists no Loss Functions for which the above Noise Source is
differentiating. Note that for $\rho_i$ described above, the synthesis algorithm
will not converge as well.
\end{theorem}
\begin{proof}
For distance metric $d_c$, a $\rho_i$ which only returns $\tup{x,
\mathsf{true}}$ is differentiating, as for all $x$, $p_a(\tup{x, \mathsf{true}})
\neq p_b(\tup{x, \mathsf{true}})$.
Note that for all Loss Functions $\cL$,
\begin{equation*}
\begin{split}
\rho_N[
\cL(\tup{x_1, \ldots x_n}, \tup{b \cdot x_1, \ldots b \cdot x_n}) - 
\cL(\tup{x_1, \ldots x_n}, \tup{a \cdot x_1, \ldots a \cdot x_n}) >
\gamma~\vert~\tup{a~\cdot x_1, \ldots a \cdot x_n}
] 
\\ + 
\rho_N[
\cL(\tup{x_1, \ldots x_n}, \tup{a \cdot x_1, \ldots a \cdot x_n}) - 
\cL(\tup{x_1, \ldots x_n}, \tup{b \cdot x_1, \ldots b \cdot x_n}) >
\gamma~\vert~\tup{b~\cdot x_1, \ldots b \cdot x_n}
] 
 \\ \leq 1
\end{split}
\end{equation*}
Therefore $\rho_N$ is not differentiating.

Similarly,
\begin{equation*}
\begin{split}
Pr[p^n_s \approx p_a~\vert~p_a] + Pr[p^n_s \approx p_b~\vert~p_b] \leq 1
\end{split}
\end{equation*}
Therefore, if for any $n$, 
$Pr[p^n_s \approx p_a~\vert~p_a]
 \geq (1 - \delta)$ then 
$
Pr[p^n_s \approx p_b~\vert~p_b] \leq \delta$
\end{proof}

\begin{theorem}\label{thm:necessitydifferentiating-A}
    Consider a Loss Function which checks if the first character appended to a
    string is either 
$``a"$ or
$``b"$.
    Given the above Loss Function and DL-2 Distance metric $d_{DL2}$, 
    the Noise Source described above is differentiating.

In this case,
if we pick a differentiating Input Source, our synthesis algorithm
will have convergence guarantees.
\end{theorem}
\begin{proof}
If 
$
\cL(\tup{y_1, \ldots y_n}, \tup{``aa" \cdot x_1, \ldots ``aa" \cdot x_1}) =
\infty$
 if for some $i$, $y_i = ``b" \cdot x_i$. Similarly, we can define this
in the case of 
Therefore if for two vectors $\vec{z}_1$ $\vec{z}_2$, $d_{DL2}(\vec{z}_1,
\vec{z}_2) \geq 1$, then $\cL(\vec{z}_2, \vec{y}) - \cL(\vec{z}_1, \vec{y}) =
\infty$ if $\rho_N(\vec{y}~\vert~\vec{z}) > 0$.
\[
\rho_N[
\forall \vec{z} \in Z^n.~d_{DL2}(\vec{z}, \vec{z}_h) \geq 1 \implies
\cL(\vec{z}, \vec{y}) - \cL(\vec{z}_h, \vec{y}) = \infty~\vert~\vec{z}_h] = 1
] = 1
\]

Note that if the Input Source is differentiating then the following is true.
\[
  \mathds{1}(\tup{?, \mathsf{true}} \in \vec{x})\rho_i(d\vec{x}) \geq (1 -
\delta)
\]
Note that 
\[
Pr[p^n_s \approx p_a~\vert~p_a, N] =  
  \mathds{1}(\tup{?, \mathsf{true}} \in \vec{x})\rho_i(d\vec{x}) \geq (1 -
\delta)
\]
Similarly,
\[
Pr[p^n_s \approx p_b~\vert~p_b, N] =  
  \mathds{1}(\tup{?, \mathsf{true}} \in \vec{x})\rho_i(d\vec{x}) \geq (1 -
\delta)
\]

\end{proof}


\end{document}